\documentclass[conference,onecolumn]{IEEEtran}
\ifCLASSINFOpdf
\else
\fi
\usepackage{amsmath, amsfonts, amssymb, graphicx, float}
\usepackage{amsbsy,amsthm}
\usepackage{mathrsfs,multirow}
\usepackage{bm,comment}
\usepackage{verbatim}
\usepackage{caption,subcaption}
\usepackage[bookmarks]{hyperref}
\usepackage{algorithm,algorithmic}
\usepackage{url,color}
\usepackage{latexsym}
\usepackage{dblfloatfix}
\usepackage[final]{changes}
\usepackage{cite}

\newcommand{\figref}[1]{Fig.~\ref{#1}}

\newtheorem{theorem}{Theorem}
\newtheorem*{theorem*}{Theorem}
\newtheorem{lemma}{Lemma}
\newtheorem*{lemma*}{Lemma}

\newcommand{\remove}[1]{}

\def\vv{\boldsymbol{v}}

\def\DD{\boldsymbol{D}}
\def\II{\boldsymbol{I}}
\def\LL{\boldsymbol{\mathcal{L}}}
\def\TT{\boldsymbol{T}}
\def\VV{\boldsymbol{V}}
\def\SS{\boldsymbol{S}}
\def\ZZ{\boldsymbol{Z}}
\def\HH{\boldsymbol{H}}
\def\AA{\boldsymbol{A}}
\def\BB{\boldsymbol{B}}
\def\WW{\boldsymbol{W}}
\def\PP{\boldsymbol{P}}

\def\ttheta{\boldsymbol{\theta}}
\def\GGamma{\boldsymbol{\Gamma}}
\def\KK{\boldsymbol{K}}
\def\vol{{\rm vol}}
\def\cut{{\rm cut}}

\def\AAcal{\boldsymbol{\mathcal{A}}}

\begin{document}
\title{Graph Filters and the Z-Laplacian}

\author{\IEEEauthorblockN{Xiaoran Yan}
\IEEEauthorblockA{Indiana University Network Science Institute\\
Bloomington, IN 47408\\
Email: xiaoran.a.yan@gmail.com}
\IEEEauthorblockN{Brian M. Sadler, Robert J. Drost, Paul L. Yu}
\IEEEauthorblockA{Army Research Laboratory\\
Adelphi, MD 20783}
\IEEEauthorblockN{Kristina Lerman}
\IEEEauthorblockA{Information Science Institute, University of Southern California\\
Marina Del Rey, CA 90292}
}

\maketitle

\begin{abstract}
In network science, the interplay between dynamical processes and the underlying topologies of complex systems has \replaced{led}{lead} to a diverse family of models with different interpretations. In graph signal processing, this is manifested in the form of different graph shifts and their induced algebraic systems. In this paper, we propose the unifying Z-Laplacian framework\added{,} whose instances can act as graph shift operators. As a generalization of the traditional graph Laplacian, \added{the }Z-Laplacian spans the space of all possible Z-matrices, i.e., real square matrices \replaced{with nonpositive off-diagonal entries.}{whose off-diagonal entries are nonpositive.}  We show that \added{the }Z-Laplacian can model general \replaced{continuous-time}{continuous time} dynamical processes, including information flows and epidemic spreading on a given graph. It is also closely related to general nonnegative graph filters in the discrete time domain. We showcase its flexibility by considering two 
applications.  First, we consider a wireless communications networking problem modeled with a graph, where the framework can be applied to model the effects of the underlying communications protocol and traffic.  Second, we examine a structural brain network from the perspective of \replaced{low- to high-frequency}{low to high frequency} connectivity.
\end{abstract}


%
\IEEEpeerreviewmaketitle

\section{Introduction}
As a powerful representation for many complex systems, a network models entities and their interactions via vertices and edges. In the field of network science, studies of topological structures, including those of vertex centrality and community structure, have \replaced{led}{lead} to fundamental insights into the organization and function of social, biological\added{,} and technological systems~\cite{newman2010networks, bonacich1987power, Fortunato10}. To model dynamical properties on a given network, different dynamical processes can be defined \replaced{over}{on top of} fixed topologies. Recent studies have demonstrated the fundamental interplay between dynamical operators and centrality and community structure measures~\cite{Borgatti05, ghosh_rethinking_2012, Ghosh2014KDD}.

In signal processing\added{,} we see the parallel development of graph signal processing (GSP). Starting with simple consensus problems on networks, Olfati-Saber, Fax\added{,} and Murray \cite{olfati-saber_consensus_2007} developed methods that can be used to design and analyze distributed control of sensors, unmanned vehicles\added{,} and communications systems. Capable of dealing with directed networks, switching topologies\added{,} and time delays, their models are closely related to random walks on graphs. Sandryhaila and Moura \cite{sandryhaila_discrete_2013} proposed a more general framework for defining linear invariant filters based on graph shift operators. By generalizing the classical discrete signal processing framework to graph topologies, well-developed theories and techniques can be extended to analyze more-complex problems involving interconnected systems and relational data sets \cite{shuman_signal_2012}.

By connecting these ideas from signal processing and network science, we investigate the mathematical duality between random walk and consensus processes on networks. We propose a class of discrete graph shift operators that are capable of modeling dynamical processes\added{,} including information flows and epidemic spreading. We demonstrate that\added{,} in this framework, these shifts span the space of all nonnegative matrices. We then adapt the discrete filters to \replaced{continuous-time}{continuous time} settings, introducing the \emph{Z-Laplacian} with heterogeneous time delays. This theoretical generalization of the \emph{parameterized Laplacian} framework \added[remark={bms: [8] used to be cited here. I think Xiaoran means to cite [10] here, instead of [8].  In the earlier version we said our approach was compatible with the general framework of [8], but the parameterized Laplacian is from [10].}]{\cite{yan2016capturing}} opens the door to signal processing analysis of continuous dynamical systems 
on networks. 

The cross-disciplinary connection also allows us to apply the idea of dynamics modeling to signal processing problems. Under the \replaced{parameterized}{parametrized} Laplacian framework \cite{yan2016capturing}, we demonstrated that different parameterizations of random walk and consensus processes can lead to different perceived network structures on the same topology. In this paper, we make the following novel contributions:
\begin{itemize}
 \item We propose the general \replaced{continuous-time}{continuous time} Z-Laplacian framework.  The associated shifts span the space of \emph{Z-matrices}, which are \deleted{mathematically defined as }real square matrices \replaced{with nonpositive off-diagonal entries}{whose off-diagonal entries are nonpositive} \cite{fiedler2008special}. \replaced{The framework}{It} enables the modeling of epidemic and information diffusion (Theorems \ref{the:unification},\ref{the:general},\ref{the:general_cont}), unifying many existing linear operators in the literature.
 \item We connect \replaced{discrete- and continuous-time}{discrete and continuous time} dynamical processes to the \replaced{GSP}{graph signal processing} framework. In particular, we propose the \emph{Z-Laplacian} operators as graph shifts, leading to induced signal processing techniques with corresponding dynamical process interpretations.
 \item We provide two signal processing examples of how different graph shift choices can lead to different conclusions in real applications. 
 \begin{itemize}
    \item For a wireless communications network with a fixed topology, we use the framework to model the traffic patterns under different communications protocols, coupling \replaced{GSP}{graph signal processing} with underlying protocol strategies and enabling the study of their interplay. This example illustrates how \deleted{the }a dynamical process on a graph \replaced{can be}{is} altered depending on the underlying assumptions about the traffic, communications rate, and protocol.
    \item For structural brain networks, we use the framework to conduct frequency analysis from different dynamical perspectives, including information diffusion models made possible by the Z-Laplacian framework. \deleted[remark=bms: This is already well known and not a novel contribution.  Maybe we should drop this sentence.]{This example also illustrates how signal processing tools are useful in network analysis.}
 \end{itemize}
\end{itemize}

\section{Background}
\label{sec:background}
Classical discrete signal processing provides a wide range of tools to analyze data on regular structures, including filtering, transformation, compression, etc.\added{,} GSP applies these tools to signals on graphs with arbitrary topologies \cite{shuman_signal_2012}. 

Consider a directed graph $G=\{V,E,\AA\}$, where $V=\{v_1,v_2,...,v_N\}$ is the set of vertices, representing $N$ elements in the system, and $E$ is the set of edges that represent the \replaced{pairwise}{pair-wise} interactions between the vertices. The topological structure of the system is captured by the weighted adjacency matrix $\AA$. \added{(Throughout, we will refer to weighted adjacency matrices simply as adjacency matrices.) }The diagonal \replaced{in- and out-degree}{in and out degree} matrices are respectively defined \replaced{by}{as} \replaced{$[\DD_\text{out}]_{uu} = \sum_v \AA_{uv}$}{$[\DD_{out}]_{uu} = \sum_v \AA_{uv}$} and \replaced{$[\DD_\text{in}]_{vv} = \sum_u \AA_{uv}$}{$[\DD_{in}]_{vv} = \sum_u \AA_{uv}$}, where $\AA_{uv}$ is the $(u,v)$ element of matrix $\AA$. For undirected graphs, \replaced{$\DD_\text{in} = \DD_\text{out} = \DD$}{$\DD_{in} = \DD_{out} = \DD$}. 

We define a \emph{graph signal} $\theta : V\rightarrow \mathbb{R}$ as a mapping from $V$ to the real numbers.  We represent the graph signal at time step $n$ (or $t$ for continuous time) as a row vector $\ttheta(n)$,\footnote{In this paper, we adopt the Markov process convention, i.e., using row vertex signal vectors $\ttheta(n)$ and \replaced{right multiplying}{right-multiply} them by matrix operators, which contrasts with the algebraic convention we used in \cite{Ghosh2014KDD,yan2016capturing}. We will also use $a_{uv}=[\AA]_{uv}$ to represent entry $(u,v)$ of \replaced{a}{the} matrix $\AA$\added{ and $d_u = [\DD]_{uu}$ to represent entry $(u,u)$ of a diagonal matrix $\DD$}.} making the space of graphs signals identical to $\mathbb{R}^N$. A \emph{graph filter} \deleted{$h$ }is a mapping from one graph signal to another,
\begin{equation}
\ttheta(n+1) = \ttheta(n) \HH\;,
\end{equation}
where the filter is represented by an $N\times N$ matrix $\HH$. Moreover, just like in classical signal processing, any linear shift-invariant filter can be \replaced{expressed}{defined} as
\begin{align}
  \HH = \HH(\SS) = h_0 \II + h_1 \SS + h_2 \SS^2 + ... + h_l \SS^l\;,
\end{align}
where \replaced{$\SS$ is known as the graph shift operator corresponding to $\HH$, the $h_i$, $0\leq i\leq l$, are real coefficients, and $l$ is the order of the filter.}{$\SS$ is the graph shift operator and $h_0, h_1,... h_l$ are real coefficients and $l$ is the order of the filter.}  \added{The graph shift operator }$\SS$ is not only the building block of shift-invariant filters, it is also closely related \replaced{to}{with} the \replaced{notions}{notion} of frequency response, convolution\added{,} and Fourier transforms on graphs \cite{sandryhaila_discrete_2014}. In \cite{sandryhaila_discrete_2013}, Sandryhaila and Moura derived a formal algebra based on $\SS$, generalizing corresponding signal processing concepts to graph topologies.

There are at least two major definitions of the shift operator based on the graph adjacency matrix $\AA$ \replaced{or}{and} the (unnormalized) Laplacian matrix $\LL = \DD-\AA$,\deleted{ respectively,} and alternatives with other properties have also been proposed \cite{gavili_shift_2015}. Because different shift operator definitions lead to divergent tools and algorithms, practitioners face difficult choices when applying \replaced{GSP}{graph signal processing} techniques. It is thus crucial to develop a basic understanding of how graph shift operators differ and relate to each other. For this purpose, we connect to the ideas from the \replaced{parameterized}{parametrized} Laplacian framework \cite{yan2016capturing}, where different operators can be interpreted as variants of random walk and consensus processes.  We first reintroduce the framework in signal processing notation, described below and listed in Table \ref{tab:glossary}.
\begin{table}
\caption{Glossary of terms and notation}
\begin{minipage}{\textwidth}
\centering
\setlength{\tabcolsep}{3.3em}
\bgroup\def\arraystretch{1.2}
	\begin{tabular}{cl}
	\hline\hline
	Term & Description \\
	\hline
	Nonnegative matrix & A real matrix with all nonnegative entries\\	
	Z-matrix & A real matrix with all \replaced{off-diagonal}{non-diagonal} entries $\le 0$\\
	$\ttheta(t)$ &Graph signal at continuous time $t$ \\
	\replaced{$\ttheta^\text{CON}(t)$}{$\ttheta^{CON}(t)$} &Graph signal under consensus basis\\
	$\PP$ &Random walk operator \\
	\replaced{$\PP^\text{CON}$}{$\PP^{CON}$} &Consensus operator \\
	$\AA$ & Adjacency matrix of graph $G$ \\		
	\replaced{${\DD}_\text{out}$}{${\DD}_{out}$} &Diagonal degree matrix of $\AA$ \\
	$\WW$ & Transformed adjacency matrix of $\AA$ \\	
	\replaced{${\DD_{\WW}}_\text{out}$}{${\DD_{\WW}}_{out}$} &Diagonal degree matrix of $\WW$ \\
	$\VV_{\AA}$ & Diagonal matrix \replaced{w/}{with the} dominating eigenvector\added{ of $\AA$}\\
	$\ZZ $ & Diagonal replicating factor matrix\\
	$\TT $ & Diagonal delay factor matrix\\
	$\LL$	& General Laplacian operator (examples follow)\\
	Random walk Laplacian	& \replaced{$\II - \DD_\text{out}^{-1}\AA$}{$\II - \DD_{out}^{-1}\AA$}, given $\AA$\\
	Parameterized Laplacian & \replaced{$\TT^{-1}(\II - {\DD_{\WW}}_\text{out}^{-1}\WW)$}{$\TT^{-1}(\II - {\DD_{\WW}}_{out}^{-1}\WW)$}, given $\WW = \AA\BB$\\
	Replicator & Parameterized Laplacian with $\WW = \VV_{\AA}\AA\VV_{\AA}$\\
	Z-Laplacian & \replaced{$\TT^{-1}(\II - \ZZ \DD_\text{out}^{-1}\AA)$}{$\TT^{-1}(\II - \ZZ \DD_{out}^{-1}\AA)$}, given $\AA$\\
	\hline\hline
	\end{tabular}\egroup
\end{minipage}
\label{tab:glossary}
\end{table}

We start by representing a \replaced{discrete-time}{discrete time} random walk as a signal on a directed graph $G=\{V,E,\AA\}$:\deleted[remark=rjd:modified equation]{}
\begin{align}
\label{eq:RW}
\ttheta(n+1) &= \ttheta(n) \PP = \ttheta(n) \DD_\text{out}^{-1}\AA\;.
\end{align}
Here the update filter $\HH = \PP$ is an $N\times N$ row (right) stochastic matrix. The graph signal $\ttheta(n)$ represents the probability density of the random walk on each vertex at step $n$. 

A consensus process on a graph can be viewed as the dual of a random walk, given by \cite{yan2016capturing},\deleted[remark=rjd:modified equation]{}
\begin{align}
\label{eq:CON}
\ttheta^\text{CON}(n+1) &= \ttheta^\text{CON}(n) \PP^\text{CON} = \ttheta^\text{CON}(n) \AA\DD_\text{in}^{-1}\;.
\end{align}
Here, the filter \replaced{$H = \PP^\text{CON}$}{$H = \PP^{CON}$} is a column (left) stochastic matrix. Assuming \replaced{$\ttheta^\text{CON}(n=0)$}{$\ttheta^{CON}(n=0)$} is the initial signal, then at every time step each vertex updates its signal using the weighted average of its neighbors via multiplication by \replaced{$\PP^\text{CON}$}{$\PP^{CON}$}. Unlike the graph signal in a random walk, entries in \replaced{$\ttheta^\text{CON}$}{$\ttheta^{CON}$} can be arbitrary (negative or positive) real numbers, without \replaced{normalization}{normalizing} constraints. 

The parameterized Laplacian $\LL$ can represent a \replaced{continuous-time}{continuous time} random walk, as in the \deleted{following }differential equation\deleted{,}\deleted[remark=rjd:modified equation]{}
\begin{equation}
\label{eq:paraP}
 \frac{d\ttheta(t)}{dt} = -\ttheta(t)\LL = -\ttheta(t) \TT^{-1}(\II - {\DD_{\WW}}_\text{out}^{-1}\WW)\;,
\end{equation}
where $\WW = \AA\BB$ is a transformed adjacency matrix\replaced{, ${\DD_{\WW}}_\text{out}$ is the diagonal matrix with $[{\DD_{\WW}}_\text{out}]_{uu} = \sum_v [\WW]_{uv} = \sum_v [\AA\BB]_{uv}$, and $\BB$ and $\TT$ are matrix parameters discussed below.  We also use $\WW$ to refer to the corresponding transformed graph itself, so that ${\DD_{\WW}}_\text{out}$ is the degree matrix of $\WW$.}{ and the degree matrix ${\DD_{\WW}}_{out}$ is also defined accordingly as: $[{\DD_{\WW}}_{out}]_{uu}= \sum_v [\AA\BB]_{uv}$ (see Table 1).}

Compared with the random walk Laplacian \replaced{$\LL = \II-\DD_\text{out}^{-1}\AA$}{$\LL = \II-\DD_{out}^{-1}\AA$}, \deleted{Equation }\eqref{eq:paraP} has two additional parameter sets, $\BB$ and $\TT$. The diagonal matrix $\BB$ consists of vertex bias factors that alter the random walk trajectory by giving neighbors additional weights.  In a biased random walk, the transition probability from vertex $u$ to $v$\added{,} denoted \replaced{$P_{uv}^\text{BRW}$,}{$P_{uv}^{BRW}$} is multiplied by a target bias factor $b_v$. In the parameterized Laplacian framework \cite{yan2016capturing,yan_multi-layer_2017}, we introduced the idea of the \emph{bias transformation} to relate the biased random walk to an unbiased version.

\begin{lemma}[Bias transformation]
\label{th:transformB}
Any biased random walk on $G = (V,E,\AA)$, with the diagonal matrix $\BB$ specifying vertex bias factors $b_v$, is equivalent to an unbiased random walk on the transformed graph $\WW = \AA \BB$.  If $G$ is undirected, \added{then we instead consider }the transformed graph \deleted{is }$\WW = \BB \AA \BB$\added{ to maintain edges having equal weight in both directions}.
\end{lemma}
\begin{proof}
 See appendix.
\end{proof}

The other diagonal matrix \replaced{parameter, $\TT$, effects time delays for the continuous-time random walk,}{$\TT$ controls the time delay of a continuous-time random walk}\footnote{Bias transformation (also called ``reweighing transformation'' in \cite{yan2016capturing}) applies to both \replaced{discrete- and continuous-time}{discrete and continuous time} dynamical processes\deleted{. See }\cite{lambiotte_laplacian_2008}.}\replaced{ providing inverse clock rates that control how long the walk stays at each vertex.}{, or inverse clock rate at which the random walk stays at each vertex.} \added{(}Without loss of generality, we constrain all the \added{diagonal} entries \replaced{$\tau_u = [\TT]_{uu}\geq 1$}{in $\TT$ with $\tau_u \geq 1$}.\added{)} In Section \ref{sec:continuous}, we will \replaced{justify}{demonstrate} this intuition by connecting \replaced{continuous-time}{continuous time} processes to their discrete counterparts. Delayed continuous-time random walks can be captured using the \emph{delay 
transformation}:
\begin{lemma}[Delay transformation]
\label{th:transformD}
Any unbiased \replaced{continuous-time}{continuous time} random walk on $G = (V,E,\AA)$, with the diagonal matrix $\TT$ specifying vertex delay factors $\tau_v$, is equivalent to a \replaced{continuous-time}{continuous time} random walk with \deleted{the }delay factors $\II$ on the transformed graph \replaced{$\WW = \DD_\text{out}(\TT- \II)+\AA$}{$\WW = \DD_{out}(\TT- \II)+\AA$}, where $\II$ is the identity matrix.
\end{lemma}
\begin{proof}
 See appendix.
\end{proof}

Delay transformation enables us to view delay factors as self-loops, which can be absorbed into $\WW$. A \deleted{simple }special case is when $\TT = \alpha \II$ is a scalar matrix, which can be understood as rescaling the global clock rate, so that all delays are identical and equal to $\alpha$.

Beyond the bias and delay transformations, the full parameterized Laplacian framework also has a similarity transformation that unifies the random walk and consensus processes on undirected graphs.
\begin{lemma}[Similarity transformation]
\label{th:transformS}
Any \replaced{continuous-time}{continuous time} random walk on an undirected graph $G = (V,E,\AA)$\deleted{,} captured by the parameterized Laplacian $\LL = \TT^{-1}(\II - \DD_{\WW}^{-1}\WW)$\added{,} with the diagonal matrices $\BB$ and $\TT$ specifying vertex bias factors and vertex delay factors, is equivalent to a \replaced{continuous-time}{continuous time} dynamical process captured by the parameterized Laplacian $\LL = (\TT\DD_{\WW})^{-1+\rho}(\DD_{\WW} - \WW)(\TT\DD_{\WW})^{-\rho}$, up to a change of basis, \replaced{where}{and} $\rho$ is the basis parameter and $0\le\rho\le1$\added[remark={bms: We don't say what $\rho$ is until the next paragraph}]{ (described below)}.
\end{lemma}
\begin{proof}
 See the ``similarity transformation'' in \cite{yan2016capturing}.
\end{proof}

\added{In particular, we recover the \emph{random walk basis} by setting $\rho = 0$, and the \emph{consensus basis} with $\rho = 1$. Another relevant case is the \emph{symmetric basis} with $\rho = 0.5$, which leads to a Laplacian operator $\LL^\text{SYM}$ represented by a symmetric matrix. In linear algebra, similarity is an equivalence relation for square matrices \cite{fiedler2008special}. Similar matrices share many key properties, including their rank, determinant\added{,} and eigenvalues. Eigenvectors are also equivalent under a change of basis. For a given initial signal on an undirected graph with $\TT = \II$ and $\WW=\AA$,} \deleted{and it follows that given the same initial signal and an undirected graph, }the random walk \eqref{eq:RW} and consensus process \eqref{eq:CON} become identical at every time step\added{,} up to a change of basis. This follows from\deleted[remark=rjd:modified equation]{}
\begin{align}
\DD\LL\DD^{-1} &= \DD(\II-\DD^{-1}\AA) \DD^{-1}\nonumber\\
&=\II - \AA\DD^{-1} = \II - \PP^\text{CON} = \LL^\text{CON}\;,
\end{align}
where we used the fact that \deleted{in undirected graphs, }$\AA = \AA^T$ and \replaced{$\DD_\text{in} = \DD_\text{out} = \DD$ for undirected graphs.}{$\DD_{in} = \DD_{out} = \DD$.}

Using bias, delay\added{,} and similarity transformations, the parameterization Laplacian framework unifies various linear operators and their associated centrality and community structure measures from the network science literature \cite{yan2016capturing}. In particular, here we introduce one special operator called the \emph{replicator}, which is related to epidemic models.

\begin{lemma}[Replicator operator]
\label{th:replicator}
If $G = (V,E,\AA)$ is undirected\deleted{,} and \replaced{$\vv_{\AA}$ is the eigenvector of the adjacency matrix $\AA$ associated with the largest eigenvalue $\lambda_\text{max}$}{$\lambda_{max}$ is the largest eigenvalue of the adjacency $\AA$ associated with the eigenvector $\vv_{\AA}$}, so that \replaced{$\vv_{\AA}\AA=\lambda_\text{max}\vv_{\AA}$}{$\vv_{\AA}\AA=\lambda_{max}\vv_{\AA}$}, a biased random walk with the diagonal matrix $\BB = \VV_{\AA}$, whose bias factors are the components of $\vv_{\AA}$\added{,}\footnote{\replaced{The replicator}{Replicator} operator, not to be confused with a ``replicating factor'', also defines the \replaced{maximum-entropy}{maximum entropy} random walk under the random walk basis \cite{gomez-gardenes_entropy_2008}.\deleted{rjd: Is this footnote attached in the best place?}}\deleted{,} is defined by the stochastic matrix \replaced{$\PP_{\WW}^\text{SYM} = \frac{1}{\lambda_\text{max}}\AA$}{$\PP_{\WW}^{\text{\textit{SYM}}} = \frac{1}{\lambda_{max}}\AA$} under the symmetric 
basis.\deleted[remark={bms: Not sure we define what ``the symmetric basis'' is...}]{}
\end{lemma}
\begin{proof}
By Lemma \ref{th:transformB}, the stochastic matrix of a biased random walk with $\BB = \VV_{\AA}$ is 
$$\PP_{\WW} = \DD_{\WW}^{-1} \VV_{\AA}\AA\VV_{\AA}\;,$$
where $\DD_{\WW}$ is the diagonal degree matrix of the transformed graph $\WW = \VV_{\AA}\AA\VV_{\AA}$. Because \replaced{${d_{\WW}}_i = \sum_j {\vv_{\AA}}_i a_{ij}{\vv_{\AA}}_j =  {\vv_{\AA}}_i\sum_j  a_{ij}{\vv_{\AA}}_j = \lambda_\text{max} {\vv_{\AA}}_i^2$}{${d_{\WW}}_i = \sum_j {\vv_{\AA}}_i a_{ij}{\vv_{\AA}}_j =  {\vv_{\AA}}_i\sum_j  a_{ij}{\vv_{\AA}}_j = \lambda_{max} {\vv_{\AA}}_i^2$}, we have \replaced{$\DD_{\WW} = \lambda_\text{max} \VV_{\AA}^2$}{$\DD_{\WW} = \lambda_{max} \VV_{\AA}^2$}, and thus
\replaced{$\PP_{\WW} = \frac{1}{\lambda_\text{max}}\VV_{\AA}^{-1}\AA\VV_{\AA}\;.$}{$\PP_{\WW} = \frac{1}{\lambda_{max}}\VV_{\AA}^{-1}\AA\VV_{\AA}\;.$}
The \replaced{continuous-time}{continuous time} counterpart of $\PP_{\WW}$ is represented by the \replaced{random}{Random} walk Laplacian $\II - \PP_{\WW}$. By setting $\rho = -1/2$, $\TT = \II$\added{,} and \replaced{$\DD_{\WW} = \lambda_\text{max} \VV_{\AA}^2$}{$\DD_{\WW} = \lambda_{max} \VV_{\AA}^2$}, \deleted{and }according to Lemma \ref{th:transformS}\deleted{,} we have
\replaced{$\LL_{\WW}^{\text{SYM}} = \VV_{\AA}(\II-\PP_{\WW})\VV_{\AA}^{-1} = \II - \frac{1}{\lambda_\text{max}}\AA\;,$}{$\LL_{\WW}^{\text{\textit{SYM}}} = \VV_{\AA}(\II-\PP_{\WW})\VV_{\AA}^{-1} = \II - \frac{1}{\lambda_{max}}\AA\;,$}
therefore \replaced{$\PP_{\WW}^{\text{SYM}} = \frac{1}{\lambda_\text{max}}\AA$}{$\PP_{\WW}^{\text{\textit{SYM}}} = \frac{1}{\lambda_{max}}\AA$}.
\end{proof}

\replaced{In the sequel}{The rest of the paper}, we \replaced{repeatedly use}{will keep using} these transformations to design flexible operators \replaced{that yield}{with} intuitive insight. While the similarity transformation becomes obsolete as we generalize to the Z-Laplacian, bias and delay transformations remain essential in practice for interpreting and comparing models on the same topology\deleted{ that provide}, as we will show in Sections \ref{sec:sample} and \ref{sec:brain}.

Following the graph filter framework \cite{sandryhaila_discrete_2013}, we consider both \replaced{$\PP^\text{CON}$}{$\PP^{CON}$} and $\PP$ as potential graph shift operators\deleted{,} with interpretable dynamical parameters. Because both shifts and their \replaced{continuous-time}{continuous time} counterparts follow the aforementioned transformations, they form an infinite family of graph shifts on a given graph $G = (V,E,\AA)$. Note that these operators all have a dominating eigenvalue of $1$, leading to an asymptotically stationary signal $\ttheta(n=\infty)$. During this process, the total signal \replaced{(i.e., the sum of the components of $\ttheta(n)$)}{$\sum \ttheta(n)$} is always conserved under the random walk basis, preventing them from modeling non-conservative processes that grow or shrink over time.

\section{Epidemic model and nonnegative filters}
In this section, we will generalize the random walk and consensus processes to \replaced{more-general}{more general} operators, \replaced{in the process unifying}{which will unify} nonnegative linear graph filters. We begin by recalling a classic epidemic model, the susceptible-infected-susceptible (SIS) model.

\subsection{Epidemic model on a graph}
To generalize beyond conservative dynamical processes, we first redefine the classic SIS epidemic model on \added{the} graph $G=\{V,E,\AA\}$ \cite{wang_epidemic_2003} using graph signals. \replaced{The graph signal $\ttheta^\text{SIS}(n)$}{$\ttheta^{SIS}(n)$} now represents the probabilities that each vertex is infected at step $n$, given by\deleted[remark=rjd: modified equation]{}
\begin{align}
\label{eq:SIS}
\ttheta^\text{SIS}(n+1) &= \ttheta^\text{SIS}(n) \HH^\text{SIS} \nonumber\\
		   &= \ttheta^\text{SIS}(n) (\mu\AA + (1-\beta)\II)\;.
\end{align}
Here, each vertex has two states, susceptible or infected. When a vertex is \emph{susceptible}, each of its neighbors will transition to the \textit{infected} state with virus infecting probability $\mu$. Once a vertex is \emph{infected}, it will return to the \textit{susceptible} state with virus curing probability $\beta$.

An important theorem about the SIS model is that its asymptotic behavior depends on the ratio $\mu/\beta$, or the \emph{effective transmissibility} of the virus. If the effective transmissibility is above the \deleted{the }epidemic threshold, namely the inverse of the largest eigenvalue of the adjacency $\AA$, it will spread to a significant portion of the network. Otherwise, it will eventually die out.

\subsection{Epidemic model filters}
To generalize \deleted{Equation }\eqref{eq:RW}, we introduce a uniform self-replicating factor $z$ after each random walk step\added{, resulting in the following update rule and corresponding difference equation}:
\begin{align}
\label{eq:epidemic}
[\ttheta'(n+1)] &=  [\ttheta'(n)] z\PP\nonumber\\
[\ttheta'(n+1) - \ttheta'(n)] &= [\ttheta'(n)](z\PP-\II)\;.
\end{align}
Compared with $\ttheta$, the signal vector $[\ttheta']$ in \deleted{Equation }\eqref{eq:epidemic} does not necessarily sum to $1$, so it is more general and capable of modeling dynamical processes like information and epidemic spreading. 

\replaced{The difference equation in \eqref{eq:epidemic} provides intuition as to the corresponding dynamics.}{To better demonstrate the ideas and develop intuition, we have rearranged the update rule into a difference equation in Equation \eqref{eq:epidemic}.} For the uniform self-replicating factor $z$, the corresponding growth rate is actually $(z-1)$ for all vertices in the network. With $z=1$, we have no replications and recover the conservative random walk process; with $z>1$, we have an expanding process \replaced{in which}{where} the incoming probability flow is scaled by $z$ while the outgoing flow remains $1$; and \replaced{with $z<1$, we have a shrinking process.}{we have a shrinking process if $z<1$.} \replaced{Also note that in all cases $z$ is}{Notice that $z$ is also} the dominating eigenvalue of $z\PP$\added{,} corresponding to the eigenvector $\boldsymbol{1}$\replaced{, the all-ones vector.}{with all in all cases.} In practice we often restrict $z\geq 0$ so that the signal vector $[\ttheta']
$ always has a positive sum, which converges to $0$ in shrinking processes.

For undirected graphs, we can rewrite \deleted{Equation }\eqref{eq:epidemic} as\deleted{,}\deleted[remark=rjd: modified equation]{}
\begin{align}
\left[\ttheta'(n+1)\right]^{\text{SYM}} &= [\ttheta'(n)]^{\text{SYM}} \left(\frac{z}{\lambda_\text{max}}\AA\right) \nonumber\\
&= [\ttheta'(n)]^{\text{SYM}}\left(\frac{z}{\lambda_\text{max}}\AA^0 + (1-\beta)\II\right)\;,
\end{align}
where we have substituted in the replicator operator under the symmetric basis (see Lemma \ref{th:replicator}), with \replaced{$\left[\ttheta'(n+1)\right]^{\text{SYM}} = \VV_{\AA} \left[\ttheta'(n+1)\right] \VV_{\AA}^{-1}$}{$\left[\ttheta'(n+1)\right]^{\text{\textit{SYM}}} = \VV_{\AA} \left[\ttheta'(n+1)\right] \VV_{\AA}^{-1}$}, and we have replaced \replaced{$\AA = \AA^0+\frac{\lambda_\text{max}(1-\beta)}{z}\II$}{$\AA = \AA^0+\frac{\lambda_{max}(1-\beta)}{z}\II$} to match the SIS epidemic model defined on adjacency matrix $\AA^0$\added{,} as in \eqref{eq:SIS}. Here, the virus infecting probability from neighbors corresponds to \replaced{$\frac{z}{\lambda_\text{max}}$}{$\frac{z}{\lambda_{max}}$}, and the virus curing probability at an infected vertex corresponds to $\beta$. Their ratio, the effective transmissibility, is \replaced{$\frac{z}{\lambda_\text{max}\beta}$}{$\frac{z}{\lambda_{max}\beta}$}, which determines how the epidemic will spread on the adjacency matrix $\AA^0$.  


With $0\leq\beta\leq 1$, the inverse of the classic epidemic threshold, i.e., the dominating eigenvalue of $\AA^0$, is \replaced{$\frac{z+\beta-1}{z}\lambda_\text{max}$}{$\frac{z+\beta-1}{z}\lambda_{max}$}. If the effective transmissibility \replaced{$\frac{z}{\lambda_\text{max}\beta}$}{$\frac{z}{\lambda_{max}\beta}$} is greater than the threshold \replaced{$\frac{z}{(z+\beta-1)\lambda_\text{max}}$}{$\frac{z}{(z+\beta-1)\lambda_{max}}$}, or simply $z>1$, we recover the definition of an expanding process. Similarly, conservative and shrinking processes are recovered with \replaced{$z=1$ and $z<1$}{$z=1, z<1$}, respectively.

If we apply the duality between the random walk and consensus models to the non-conservative process in \deleted{Equation }\eqref{eq:epidemic}, we have\deleted[remark=rjd: modified equation]{} 
\begin{align}
\label{eq:dualUni}
[\ttheta'(n+1)] &=  [\ttheta'(n)]z \DD_\text{out}^{-1}\AA \nonumber\\
[\ttheta'(n+1)]^\text{CON} &= [\ttheta'(n)]^\text{CON}z \AA\DD_\text{in}^{-1} \;.
\end{align}
Here\added{,} the dual is a graph filter where each vertex first updates its signal using the weighted neighbor average\replaced{ and}{,} then the average signal is amplified by a factor $z$, leading to \replaced{$[\ttheta'(n=\infty)]^\text{CON} = \boldsymbol{0}$}{$[\ttheta'(n=\infty)]^{CON} = \boldsymbol{0}$} with $z<1$\deleted{,} or \replaced{$[\ttheta'(n=\infty)]^\text{CON} = \boldsymbol{\infty}$}{$[\ttheta'(n=\infty)]^{CON} = \boldsymbol{\infty}$} with $z>1$.

\subsection{General nonnegative graph filters}
To further generalize beyond epidemics with uniform self-replications, consider the following update and difference equations:
\begin{align}
\label{eq:general}
[\ttheta(n+1)] &= [\ttheta(n)] \ZZ \PP \nonumber\\
[\ttheta(n+1) - \ttheta(n)] &=  [\ttheta(n)] (\ZZ \PP -\II)\;,
\end{align}
where we have used a positive diagonal matrix $\ZZ$\added{,}\footnote{The replicating factors play the same role in the \replaced{continuous-time}{continuous time} Z-Laplacian, thus the notation $\ZZ$ here.}\deleted{,} whose diagonal elements $[\ZZ]_{vv}$ model a shrinking or expanding \emph{replicating factor} for each vertex v. The random walk step is now followed by a vertex specific replicating process, with a generally non-uniform replicating factor specified by $\ZZ$.

Applying the same duality between consensus and random walk processes to this more general dynamical process, we have\deleted[remark=rjd: modified equation]{} 
\begin{align}
\label{eq:genDual}
[\ttheta(n+1)] &= [\ttheta(n)]\ZZ \DD_\text{out}^{-1}\AA \nonumber\\
[\ttheta(n+1)]^\text{CON} &= [\ttheta(n)]^\text{CON} \AA\DD_\text{in}^{-1}\ZZ \;.
\end{align}
To interpret the operator \replaced{$\AA\DD_\text{in}^{-1}\ZZ$}{$\AA\DD_{in}^{-1}\ZZ$}, we focus on the dynamics of a specific vertex $u$, given by\deleted[remark=rjd: modified equation]{}
\begin{align}
[\ttheta(n+1)]_u^\text{CON} = \sum_v [\ttheta(n)]^\text{CON}_v [\AA]_{vu} [\DD_\text{in}]_u^{-1}[\ZZ]_{uu} \;,
\end{align}
where \replaced{$[\AA]_{vu} [\DD_\text{in}]_u^{-1}$}{$\AA_{vu} [\DD_{in}]_u^{-1}$} forms a weighted probability distribution over all incoming neighbors of $u$. Notice that the replicating factor of the incoming signals $[\ZZ]_{uu}$ only depends on the target vertex $u$. Compared with the uniform replicating factor $z$ in \eqref{eq:dualUni}, the order of matrix \replaced{multiplication by $\ZZ$}{$\ZZ$ multiplication} now matters. Under the consensus model, with vertex $u$ averaging the signals of neighbors $v$, all signals are multiplied by the same factor of $[\ZZ]_{uu}$, whereas this factor is $v$ dependent under the random walk basis.

Unlike the less general filters we have discussed previously, both operators in \eqref{eq:genDual} span the same vector space. This equivalence is easiest to show if we have an undirected graph and let\footnote{Mathematically\added{,} one has the liberty to manipulate the parameters in $\ZZ$ and $\TT$. In practice, however, we suggest setting them based on domain knowledge for intuition and interpretation.} $\ZZ = \DD$, leading to\deleted[remark=rjd: modified equation]{}
\begin{align}
\AA\DD_\text{in}^{-1}\ZZ = \AA\DD^{-1}\DD = \DD\DD^{-1}\AA = \ZZ \DD_\text{out}^{-1}\AA \;,
\end{align}
where we used $\AA = \AA^T$\deleted{,} and \replaced{$\DD=\DD_\text{in}=\DD_\text{out}$}{$\DD=\DD_{in}=\DD_{out}$}.

In fact, for general directed graphs, the vector space spanned by both operators contains all possible nonnegative matrices, and we call them \emph{general nonnegative filters}. To prepare for this theoretical unification, we first consider the following lemma regarding the adjacency matrix and random walks.

\begin{lemma}[Adjacency mapping]
\label{the:mapping}
For every directed weighted graph $G = \{V,E,\AA\}$, 
  there is a unique transition matrix, \replaced{$\PP_{\AA} = \DD_\text{out}^{-1}\AA $}{$\PP_{\AA} = \DD_{out}^{-1}\AA $},
  that captures an unbiased random walk on $\AA$.
Conversely, given a stochastic matrix $\PP$, there is an infinite family of adjacency matrices $\AAcal_{\PP}$ whose random walk process
  is consistent with $\PP$, denoted as
\[
\AAcal_{\PP} = \{\GGamma \PP : \mbox{$\GGamma$ is a positive diagonal matrix}.\}
\]
\end{lemma}
\begin{proof}
 Since \replaced{$\DD_\text{out}$}{$\DD_{out}$} is uniquely determined by a given $\AA$, then every directed network uniquely defines a random walk process. However, given a transition matrix $\PP$, there remains $N$ degrees of freedom to specify the underlying network. Intuitively, the degrees of freedom can be interpreted as row scalings that proportionally multiply all outgoing edges of a vertex, thus preserving the random walk distribution leaving from the given vertex. Here we represent these degrees of freedom using the entries of $\GGamma$.
\end{proof}

Now we are ready to prove the unification theorem for nonnegative filters expressed with an arbitrary basis\replaced{.}{,}
\begin{theorem}[Basis Unification]
\label{the:unification}
For any general nonnegative filter defined as \replaced{$\ZZ \DD_\text{out}^{-1}\AA$}{$\ZZ \DD_{out}^{-1}\AA$} on graph $G = \{V,E,\AA\}$, there is an equivalent dual filter defined as $({\DD'}^{-1}\ZZ')^{1-\rho}\AA'({\DD'}^{-1}\ZZ')^{\rho}$ on a dual graph $\AA' \in \AAcal_{\PP}$, under any given basis parameter $0\le\rho\le 1$.
\end{theorem}
\begin{proof}
Assume we have two equivalent filters defined on two different graphs $\AA$ and $\AA'$. Then\deleted[remark=rjd: modified equation]{}
\begin{align}
  \ZZ \DD_\text{out}^{-1}\AA &= ({\DD'}^{-1}\ZZ')^{1-\rho}\AA'({\DD'}^{-1}\ZZ')^{\rho}\nonumber\\
   \AA' &=  ({\DD'}^{-1}\ZZ')^{\rho} \ZZ \PP_{\AA} ({\DD'}^{-1}\ZZ')^{-\rho}\;.
\end{align}
Setting $\ZZ' = {\DD'}$, we have $\AA'= \ZZ \PP_{\AA}$. Applying Lemma \ref{the:mapping}, the dual graph $\AA'$ is a row scaled version of $\AA$ by setting $\GGamma=\ZZ$.
\end{proof}
By Theorem \ref{the:unification}, the duality in \eqref{eq:genDual} essentially leads to another filter on a different graph. When \replaced{$\rho = 1$ and $\DD' = {\DD_{\AA'}}_\text{in}$}{$\rho = 1, \DD' = {\DD_{\AA'}}_{in}$}, we have the equivalence between the consensus and random walk processes. As a result, we will no longer need to separately specify dynamical processes in random walk, consensus\added{,} or any other basis. We will thus suppress the superscript notations on filters and signals in the following sections.

Next, we prove that general nonnegative filters span all possible nonnegative matrices:
\begin{theorem}[Generality Theorem]
\label{the:general}
Given an arbitrary nonnegative matrix $\HH$, we can associate it with a general nonnegative filter \replaced{$\HH = \ZZ \DD_\text{out}^{-1}\AA$}{$\HH = \ZZ \DD_{out}^{-1}\AA$}, which models a general dynamical process on a graph $G = \{V,E,\AA\}$.
\end{theorem}
\begin{proof}
 Let $\AA = \HH$. For any nonnegative matrix $\HH$, $\AA = \HH$ represents a graph adjacency matrix. By setting \replaced{$\ZZ = \DD_\text{out}$}{$\ZZ = \DD_{out}$}, we have \replaced{$\HH = \ZZ \DD_\text{out}^{-1}\AA$}{$\HH = \ZZ \DD_{out}^{-1}\AA$}.
\end{proof}
The simple proof essentially states that any nonnegative matrix can be interpreted as an epidemic model based on a random walk with non-uniform replicating factors that are proportional to the vertex degrees. From a degrees-of-freedom perspective, the non-uniform replicating factor $\ZZ$ perfectly matches the row normalization constraint.

Combining \replaced{Theorems \ref{the:unification} and \ref{the:general}}{Theorem \ref{the:unification} and Theorem \ref{the:general}}, we observe the full generality of the proposed nonnegative graph filters. Based on the algebraic framework introduced in \cite{sandryhaila_discrete_2013}, \deleted{using }these nonnegative shift operators also define\deleted{s} a shift-invariant vector space, which can be generally expressed as $\mathcal{F} = \left\{\KK: \KK=\sum_{l=0}^N h_l \HH^l \mid h_l\in \mathbb{R}\right\}$. This unification opens the door to discrete signal processing methodologies while connecting the concepts of random walk, epidemics\added{,} and information diffusion on networks. 

\section{\replaced{Continuous-time}{Continuous time} filters and Z-Laplacian}
\label{sec:continuous}
Based on the discrete nonnegative graph filters in \deleted{Equation }\eqref{eq:general}, we can define continuous dynamical processes by replacing the difference equations with differential equations, \replaced{yielding}{given by}\footnote{The variable $t$ is a real number representing continuous time.}\deleted[remark=rjd: modified equation]{}
\begin{align}
\label{eq:general_cont}
\frac{d \ttheta(t)}{dt} = -\ttheta(t)\LL =  -\ttheta(t) \TT^{-1}(\II - \ZZ \DD_\text{out}^{-1}\AA)\;,
\end{align}
where we define the operator \replaced[remark={bms: We include the $T^{-1}$ term out front in the definition in Table I and also in Lemma 6}]{$\TT^{-1}(\II - \ZZ \DD_\text{out}^{-1}\AA)$}{$\II - \ZZ \DD_{out}^{-1}\AA$} as the \emph{Z-Laplacian}. 

Here the positive diagonal matrix $\TT$ represents the time delay, or inverse clock rate, at each vertex. When $\ZZ=\II$, we have $\LL = \TT^{-1}(\II-\PP)$, which corresponds to a special case of the \replaced{parameterized}{parametrized} Laplacian we introduced in~\cite{Ghosh2014KDD}. Compared with \replaced{discrete-time}{discrete time} filters with uniform time steps, \deleted{Equation }\eqref{eq:general_cont} allows asynchronous updates by modeling the movement of random walks as Poisson processes for which the waiting times between jumps are exponentially distributed with \replaced{mean parameters}{means} specified by $\TT$ \cite{lambiotte_laplacian_2008}.

To connect the \replaced{continuous-time}{continuous time} Z-Laplacian to discrete filters, we first prove the following Lemma over a short time interval $\delta$\replaced{.}{,} 

\begin{lemma}[Discrete approximation]
\label{th:discrete-continuous}
Given the graph signal $\ttheta(t)$ and the \replaced{continuous-time}{continuous time} Z-Laplacian \replaced{$\TT^{-1}(\II - \ZZ \DD_\text{out}^{-1}\AA)$}{$\TT^{-1}(\II - \ZZ \DD_{out}^{-1}\AA)$}, the graph signal at time $t+\delta$ can be approximated as $\delta\rightarrow 0$ by\deleted{,}
$$\lim_{\delta\rightarrow 0}\ttheta(t+\delta) =\ttheta(t) e^{-\LL\delta }  = \ttheta(t)(\II-\lim_{\delta\rightarrow 0}\delta\LL)\;,
$$
where $\II-\delta\LL$ represents a discrete time filter.
\end{lemma}
\begin{proof}
 See appendix.
\end{proof}
\replaced{Letting}{Let} $\delta = \frac{1}{\lambda}$, we rewrite the discrete time filter as\deleted[remark=rjd: modified equation]{}
\begin{align}
\boldsymbol{\varPhi} = \II - \frac{1}{\lambda}\LL = \II - \frac{1}{\lambda}\TT^{-1}(\II - \ZZ \DD_\text{out}^{-1}\AA)\;.
\end{align}
Substituting $\LL = \lambda(\II-\boldsymbol{\varPhi})$, we can write the solution of \eqref{eq:general_cont} as\deleted[remark=rjd: modified equation]{}
\begin{align}
\ttheta(t) = \ttheta(0) e^{-\LL t} = \ttheta(0) e^{-(\II-\boldsymbol{\varPhi})\lambda t}= \ttheta(0) e^{-\lambda t}e^{\lambda \boldsymbol{\varPhi} t}\;.    
\end{align}
\replaced{Using the}{The} ``uniformization technique'' \cite{reibman1988numerical} \replaced{we can }{allows us to} write matrix exponentiation in terms of \added{an }infinite sum of matrix powers\replaced{ to obtain}{,}\deleted[remark=rjd: modified equation]{}
\begin{align}
\label{eq:uniform}
\ttheta(t) = \ttheta(0) e^{-\lambda t} \sum_{k=0}^{\infty} \frac{1}{k!} (\lambda \boldsymbol{\varPhi} t)^k &= \ttheta(0)\sum_{k=0}^{\infty} \frac{ e^{-\lambda t}}{k!} (\lambda \boldsymbol{\varPhi} t)^k \nonumber\\
	   &= \ttheta(0) \sum_{k=0}^{\infty} P_\text{Poi}(k,t) \boldsymbol{\varPhi}^k \;,
\end{align}
where \replaced{$P_\text{Poi}(k,t)$}{$P_{poi}(k,t)$} represents a Poisson probability mass function of a \replaced{discrete-time}{discrete time} random walk \added{with intensity $\lambda$} taking $k$ steps during the time interval $[0,t)$\deleted{, with intensity $\lambda$}.

As \deleted{Equation }\eqref{eq:uniform} demonstrates, the continuous process can be interpreted as a weighted infinite sum of powers of $\boldsymbol{\varPhi}$, representing discrete time random walks of different lengths. The following theorem give us \deleted{an }intuition \replaced{for}{of} the delay factors $\TT$.

\begin{theorem}[Delay interpretation]
\label{th:interpretation}
Given the \replaced{continuous-time}{continuous time} Z-Laplacian \replaced{$\TT^{-1}(\II - \ZZ \DD_\text{out}^{-1}\AA)$}{$\TT^{-1}(\II - \ZZ \DD_{out}^{-1}\AA)$} and its discrete time approximation $\boldsymbol{\varPhi}$, the delay factor \replaced{$[\TT]_{uu}$ of vertex $u$}{of vertex $u$, $[\TT]_{uu}$} is proportional to the expected ``waiting steps'' on vertex $u$ for the approximated \replaced{discrete-time}{discrete time} random walk.
\end{theorem}
\begin{proof}
 See appendix.
\end{proof}

Lemma \ref{th:discrete-continuous} and Theorem \ref{th:interpretation} provide interpretation of the delay factors in $\TT$\deleted{,} and connect the continuous\added{-} and discrete-time models. In the next section we give an example, applying the \replaced{continuous-time}{continuous time} model to a communications network.

Next, we prove the central theorem of the Z-Laplacian framework, which \replaced{extends}{extents} the generality Theorem \replaced{(Theorem \ref{the:general})}{\ref{the:general}} to the continuous domain.
\begin{theorem}[\replaced{Continuous-Time Generality Theorem}{Continuous time generality theorem}]
\label{the:general_cont}
For an arbitrary Z-matrix $\LL$, we can associate it with a Z-Laplacian \replaced{$\LL = \TT^{-1}(\II - \ZZ \DD_\text{out}^{-1}\AA)$}{$\LL = \TT^{-1}(\II - \ZZ \DD_{out}^{-1}\AA)$}, which models a general \replaced{continuous-time}{continuous time} dynamical process on a graph (with self-loops) $G = \{V,E,\AA\}$.
\end{theorem}
\begin{proof}
Without loss of generality, we set $\TT = \delta\II$, and thus\deleted[remark=rjd: modified equation]{}
 \begin{align}
-\LL &= \frac{1}{\delta} (\ZZ\DD_\text{out}^{-1}\AA - \II)\nonumber\\
 \ZZ\DD_\text{out}^{-1}\AA &= \II-\delta\LL\;.
 \end{align}
Let $\LL$ be an arbitrary Z-matrix, i.e., a real square matrix \replaced{with nonpositive off-diagonal entries}{whose off-diagonal entries are nonpositive}\cite{fiedler2008special}. \replaced{Then}{Therefore,} $-\LL$ has only nonnegative off-diagonal entries. If there are negative diagonal entries, we set \replaced{$0\le \delta \le 1/\max_i |[\LL]_{ii}|$}{$0\le \delta \le \frac{1}{\max_i |[\LL]_{ii}|}$}, making $\II-\delta\LL$ a nonnegative matrix.

By setting \replaced{$\ZZ = \DD_\text{out}$}{$\ZZ = \DD_{out}$}, we have $\AA = \II-\delta\LL$, which represents a graph adjacency matrix given any Z-matrix $\LL$.
\end{proof}

Theorem \ref{the:general_cont} justifies the generalized \replaced{continuous-time}{continuous time} operator as the \replaced{Z-Laplacian}{\emph{Z-Laplacian}}, which is \replaced{a}{the} unifying framework for all potential shift operators. The Z-Laplacian is closely related to the \emph{generalized Laplacian} \cite{taylor_eigenvector2015, PavezGlap2016}, with the latter being the symmetric subset of the Z-Laplacian.

\section{Communications Network Analysis}
\label{sec:sample}
It is common practice to use a graph \replaced{to model}{model for} a communications network, with edges modeling one-hop connectivity between nodes. For example, convergence analysis of a consensus process as in \eqref{eq:CON} is linked with the properties of the adjacency matrix of the graph \cite{olfati-saber_consensus_2007}. This assumes an underlying communications protocol based on the connectivity, where each iteration requires every node to have a message exchange with its one-hop neighbors.  This is a useful abstraction, but does not model the impact of channel variation, rate, multi-user interference, or delay, as occur in wireless communications applications such as sensor and mobile ad-hoc networks. Employing the Z-Laplacian framework enables the modeling of delays and multi-user collisions, which are a function of the communications protocol and the graph topology.  To illustrate this, we consider analysis of a communications network to include delay effects\deleted{,} and their linkage with the 
topology and the medium access control (MAC) protocol.  We study the isolation of network bottlenecks, including the impact of MAC choices, as well as topology alteration by inserting nodes and changing bandwidth allocation to \replaced{improve network connectivity}{enhance the network and eliminate the bottleneck}.

\subsection{Networking Bottlenecks and the Z-Laplacian}
Given a network graph, consider the problem of identifying bottlenecks where traffic may funnel through a small \replaced{subset of nodes}{node subset} and create excess delays in the overall network performance. We build on prior work for bottleneck identification based solely on the network topology \cite{Sadler2,cheng2017network}, briefly reviewed next.  We use a running example throughout this section, illustrated in Figures \ref{fig:wireless}, \ref{fig:wirelessTDMA}, and \ref{fig:hybrid}. In the figures, colors indicate graph subset membership, colored disks around each vertex indicate self-loops whose weight is proportional to communications delay at that vertex \added{(}with larger area implying more delay\added{)}, and graphical edge thickness is proportional to link rate between the two vertices \added{(}with a thicker edge implying higher bandwidth\added{)}.

In the simplest setting \cite{Sadler2,cheng2017network}, we assume that all communications are orthogonal, with no multi-user interference or additional delay, and that each vertex always has a packet to transmit (a fully saturated network traffic model). Recalling our notation from Section II, the network graph is $G=\{V,E,\AA\}$, and elements of $\AA$ are denoted \replaced{$a_{ij}$}{$a_{i,j}$}. We begin with the graph network model shown in \figref{fig:wireless}.
\begin{figure}
  \begin{subfigure}[b]{0.48\textwidth}
    \centering
    \includegraphics[width=\textwidth]{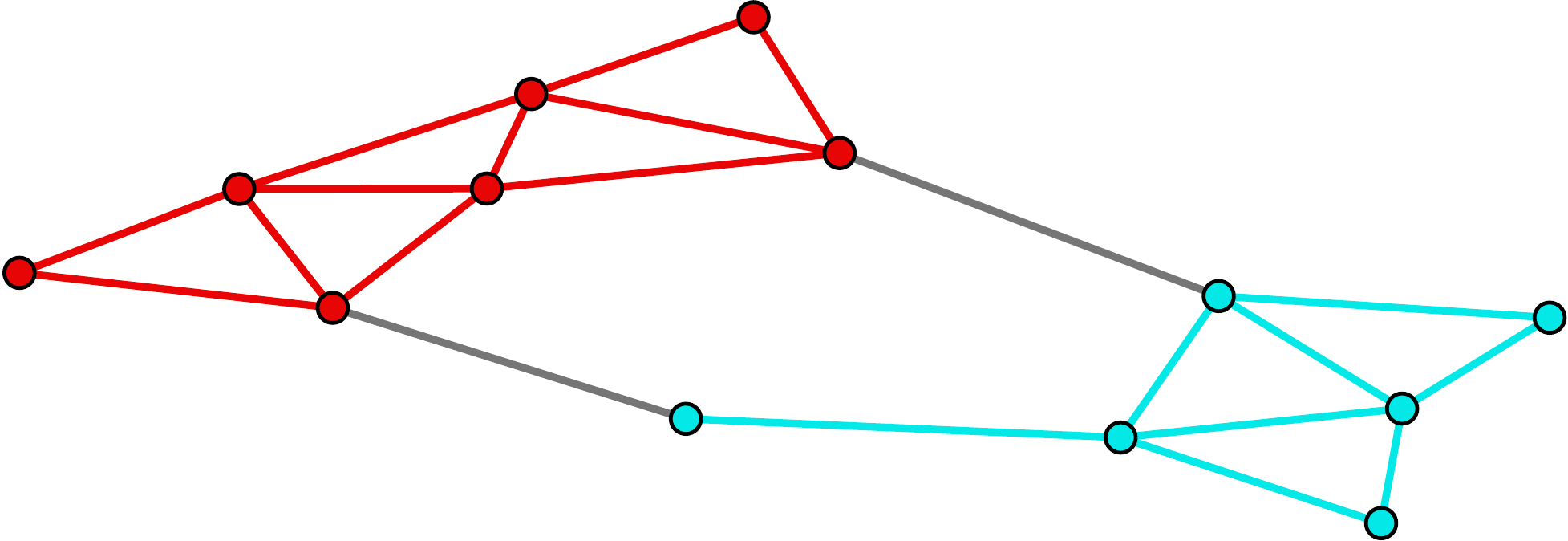}
    \caption{Traffic = $42$, Conductance = $0.22$}
  \end{subfigure}
  \begin{subfigure}[b]{0.48\textwidth}
    \centering
    \includegraphics[width=\textwidth]{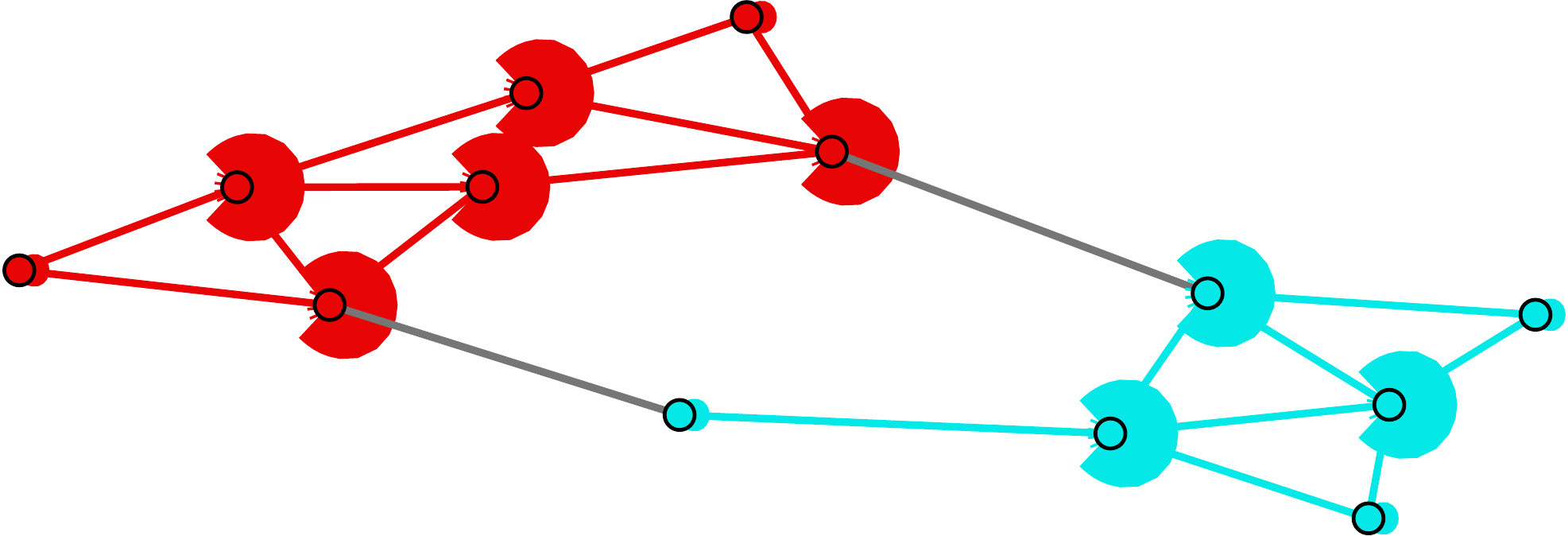}
    \caption{Traffic = $148$, Conductance = $0.067$}
  \end{subfigure}
\caption{Communications network example graph (see text). Colors indicate graph subset membership after bottleneck discovery. \replaced{(a)}{1(a):} An ideal network graph with uniform and orthogonal connectivity\deleted{,} and no delays.  \replaced{(b)}{1(b):} To model a random access MAC protocol, self-loops are introduced to model traffic delay, indicated graphically as colored circles around the vertices\replaced{, the size of the circles being}{ and whose size is} proportional to delay. In both \replaced{subfigures,}{cases} the primary communications bottleneck lies at the boundary of the two vertex communities, indicated by their different colors. The total traffic is measured by the total weighted degree, and the global efficiency is measured by the conductance.}
\label{fig:wireless}
\end{figure}
Assuming each undirected edge carries one unit of traffic in both directions, the overall network traffic equals the total weighted degree, or $2 \times |E| = 42$ for our example. Under the Z-Laplacian framework\added{,} as in \eqref{eq:general_cont}, this simple traffic pattern can be captured by setting both the vertex delays and replicating factors to the identity matrix, so $\ZZ = \TT = \II$\deleted{,} and
\begin{equation}
\label{eq:base}
 \LL = \II - \DD^{-1}\AA\;,
\end{equation}
where $\DD$ is the (diagonal) degree matrix (see Table \ref{tab:glossary}). This corresponds to the idealized case in \replaced{Fig.}{Figure} 1(a). By setting $\TT = \II$, we assume each vertex spends one unit of delay to process the packet, whereas additional delays will represent interference or collisions\replaced{, as described below}{ later on}. Here we assume that packets only flow along graph edges\deleted{,} and that each packet is only sent once (there is no duplication or broadcasting of packets), as indicated by $\ZZ = \II$.

To find and quantify the \added{primary }network bottleneck under different MAC protocols, we use the more flexible transformed graph $\WW$ for definitions (see Table \ref{tab:glossary}). For the \replaced{base-case}{base case} Laplacian in \eqref{eq:base}, we simply set $\WW = \AA$.

We split the vertices into two subsets, $S\subseteq V$ and its complement $\bar{S}= V\setminus S$. Let \replaced{$\cut(S,\bar{S})=\sum_{i \in S, j \in \bar{S}} (w_{ij}+w_{ji})$}{$\cut(S,\bar{S})=\sum_{i \in S, j \in \bar{S}} (w_{i,j}+w_{j,i})$} denote the total weights of all edges between $S$ and $\bar{S}$. Let \replaced{$\vol(S)=\sum_{i\in S} d_i = \sum_{i\in S,j\in V} w_{ij}$}{$\vol(S)=\sum_{i\in S} d_i = \sum_{i\in S,j\in V} w_{i,j}$} denote the total (undirected, weighted) degree over all vertices in S. The ratio of these two quantities measures the balanced separation strength of the transformed graph, given by 
\begin{equation}
   \phi(S)  =  \frac{\cut_{\WW}(S,\bar{S})}{{\min\left\{\vol(S),\vol(\bar{S})\right\}}}\;.
\end{equation}
The quantity $\phi(S)$ is called the \emph{conductance}, and we will use its minimum
\begin{equation}
\label{eq:conditionMin}
   \phi^*(G)  = \min_{S\in V} \phi(S) =  \phi(S^*)\;
\end{equation}
to measure the network bottleneck and its minimizer $S^*$ to locate the bottleneck at its boundary\footnote{We use the optimization algorithm introduced in \cite{yan2016capturing} to find both.}. The \replaced{base-case}{base case} Laplacian leads to a minimum conductance of $0.078$. 

\subsection{MAC Protocols and the Z-Laplacian}
The baseline analysis above assumes ideal and orthogonal communications, without any consideration of delay. Next\added{,} we adjust parameters of the Z-Laplacian to model different dynamical processes (protocols) on the same topology. Note that each protocol will induce a potentially different primary bottleneck position and strength according to \eqref{eq:conditionMin}, without changing the topology. 

Generally, depending on the protocol and topology, the vertex delays will be non-uniform. Consider a random access MAC, which will result in transmission collisions\deleted{,} and packet delays due to the need for \replaced{backoff}{back off} and retransmission. As one \replaced{modeling approach}{approach to model this}, let the vertex delays be proportional to the vertex degree, implying that \replaced{higher-degree}{higher degree} vertices will have more communications delay because collisions are more likely. Specifically, let $\TT=\DD$ so the delays are equal to the degree for each vertex\added{.}\footnote{More complicated \replaced{super-linear}{super linear} scaling can be implemented by other choices for the delay factors.}\deleted{.} The resulting Z-Laplacian is 
\begin{equation}
 \LL = \DD^{-1}(\II - \DD^{-1}\AA)\;.
\end{equation}
By appealing to the delay transformation in Lemma \ref{th:transformD}, delays are modeled by introducing vertex self-loops, and we have the corresponding transformed graph $\WW = \DD(\TT- \II)+\AA$. Applying this transformation to our example, \deleted{we have Figure 1(b), }again assuming a saturated traffic model\added{, we obtain Fig. 1(b)}. Now the more realistically modeled network is considerably less efficient than its idealized version, with total traffic $2 \times |E| = 148$\deleted{,} and a minimum conductance of $0.067$. Here, we have repeated the above bottleneck discovery algorithm, with the two resulting subsets shown in \replaced{Fig.}{Figure} 1(b). Although not evident from this example, we emphasize that changing the network protocol may result in a different bottleneck position and strength.

\replaced{Suppose now that}{Now suppose} we employ a \replaced{time-division}{time division} multiple access (TDMA) protocol to minimize collisions and enhance overall network throughput. \figref{fig:wirelessTDMA}(a) models this case as a graph filter, where each vertex evenly divides its transmission time among its neighbors. As a result, edges between \replaced{high-degree}{high degree} vertices have a reduced effective bandwidth due to their smaller share of the \replaced{time-division}{time division} allocation. This is readily captured under the Z-Laplacian framework using the bias transformation (Lemma \ref{th:transformB}): 
\begin{equation}
\label{eq:degreeBias}
 \LL = \TT^{-1}[\II - \DD_{\WW}^{-1} (\DD^{-1}\AA\DD^{-1})]\;, 
\end{equation}
where we have applied the undirected version of the bias transformation on both sides of $\AA$\deleted{,} and $\DD_{\WW'}^{-1}$ is the diagonal degree matrix of the transformed graph $\WW' = \DD^{-1}\AA\DD^{-1}$. For comparison, we started with the same traffic \added{as in \figref{fig:wireless}(a)}, with each edge carrying one unit of traffic in both directions\deleted{ as in \figref{fig:wireless} (a)}. We first saturated the time-divided effective bandwidth and then visualized the remaining traffic as self-loops or, effectively, vertex delay factors $\TT$. The resulting transformed graph is $\WW =  \DD_{\WW'}(\TT- \II)+\WW'$, as visualized in \figref{fig:wirelessTDMA}(a).
\begin{figure}
  \begin{subfigure}[b]{0.48\textwidth}
    \centering
    \includegraphics[width=\textwidth]{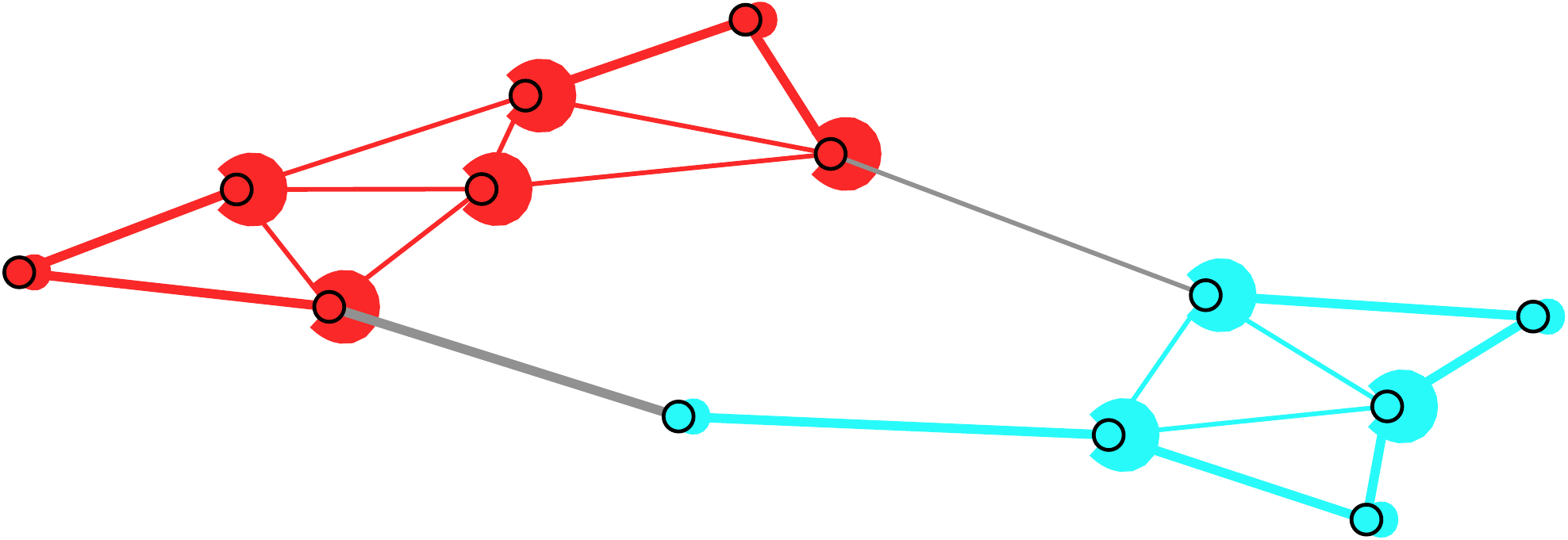}
    \caption{Traffic = $42$, Conductance = $0.078$}
  \end{subfigure}
  \begin{subfigure}[b]{0.48\textwidth}
    \centering
    \includegraphics[width=\textwidth]{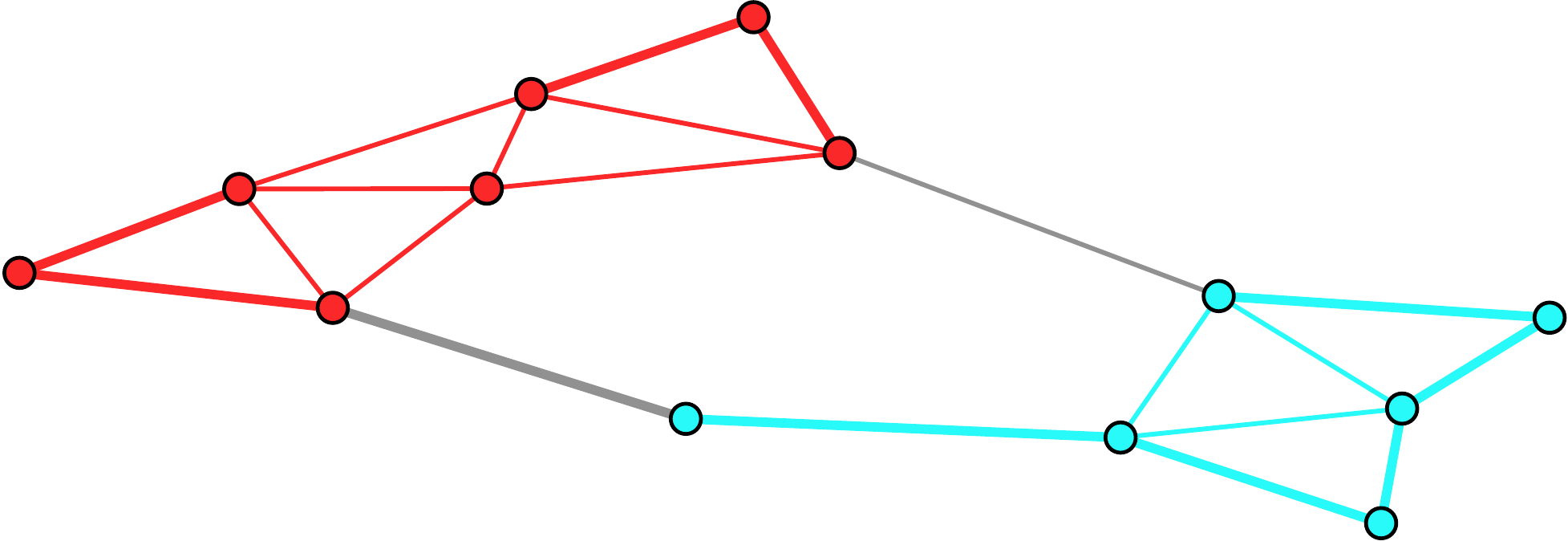}
    \caption{Traffic = $31$, Conductance = $0.21$}
  \end{subfigure}
\caption{Communications network example, continued (see text). \replaced{(a)}{2(a):} A TDMA protocol is modeled with self-loops modeling the delay time at each vertex due to time slot allocation, assuming a saturated traffic load.  \replaced{(b)}{2(b):} The same TDMA protocol is modeled, now with the traffic load perfectly matched to the available time slots, resulting in no delays.}
\label{fig:wirelessTDMA}
\end{figure}
It is a more efficient system compared \replaced{to that of }{with }\figref{fig:wireless}(b), with increased minimum conductance (now equal to $0.078$) and no redundant traffic (equal to $42$).

When the network traffic load becomes lighter, TDMA protocols are less efficient in terms of utilizing the full bandwidth. \figref{fig:wirelessTDMA}(b) models a perfectly loaded TDMA system, where the traffic is precisely matched to the capacity of each edge and each message hits its time slot without any delay. Compared with the perfectly loaded random access system in \figref{fig:wireless}(a), the network carries much less total traffic. These examples illustrate how the traffic load and protocol-induced delays can be modeled together to study the overall network performance, including the discovery of bottlenecks.

\subsection{Healing Networks by Vertex Insertion and Bandwidth Allocation}
Once a bottleneck has been identified, we can consider network healing by augmenting the network. There are various ways to do this, such as bandwidth reallocation to reduce delays. Here\added{,} we consider augmenting the network topology by introducing a (perhaps higher bandwidth) link across the bottleneck. We limit our study to the introduction of a new edge in the existing graph, although the general framework allows for more elaborate combinations of new vertex introduction and optimized resource allocation, which can correspond to an adaptive enhancement at the network physical layer and/or the MAC. There are many possibilities\added{,} and this is an interesting topic for future study. 

We continue our \added{running }example in this section, beginning with the random access protocol model in \figref{fig:wireless}(b). \replaced{Figs. \ref{fig:hybrid}(a) and \ref{fig:hybrid}(b)}{\figref{fig:hybrid}(a) and (b)} depict two cases where a new edge is introduced.
\begin{figure}
  \begin{subfigure}[b]{0.48\textwidth}
    \centering
    \includegraphics[width=\textwidth]{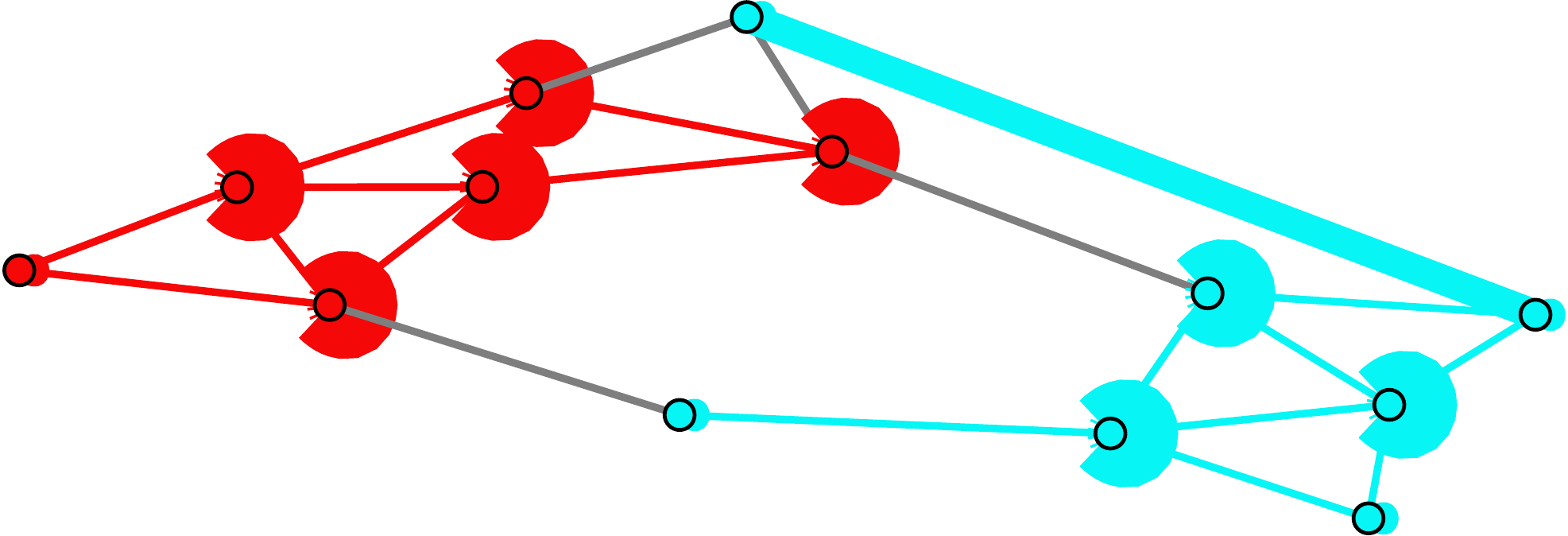}
    \caption{Traffic = $156$, Conductance = $0.11$}
  \end{subfigure}
  \begin{subfigure}[b]{0.48\textwidth}
    \centering
    \includegraphics[width=\textwidth]{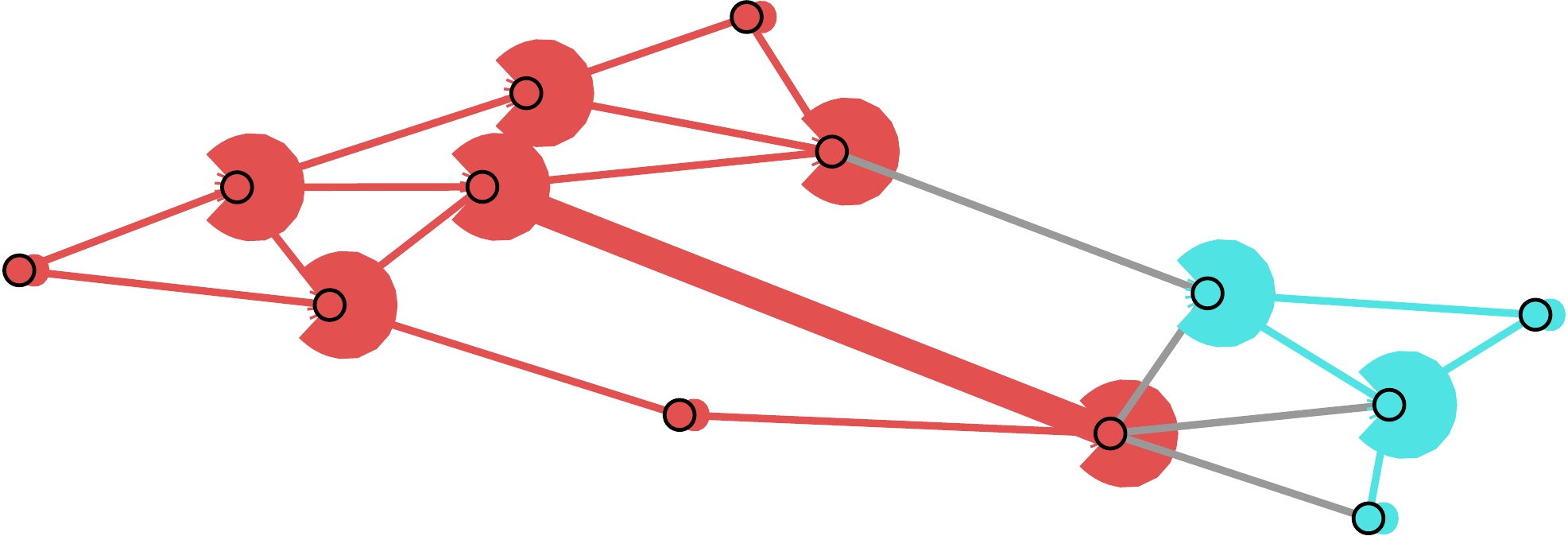}
    \caption{Traffic = $156$, Conductance = $0.2$}
  \end{subfigure}
  \begin{center}
  \begin{subfigure}[b]{0.48\textwidth}
    \centering
    \includegraphics[width=\textwidth]{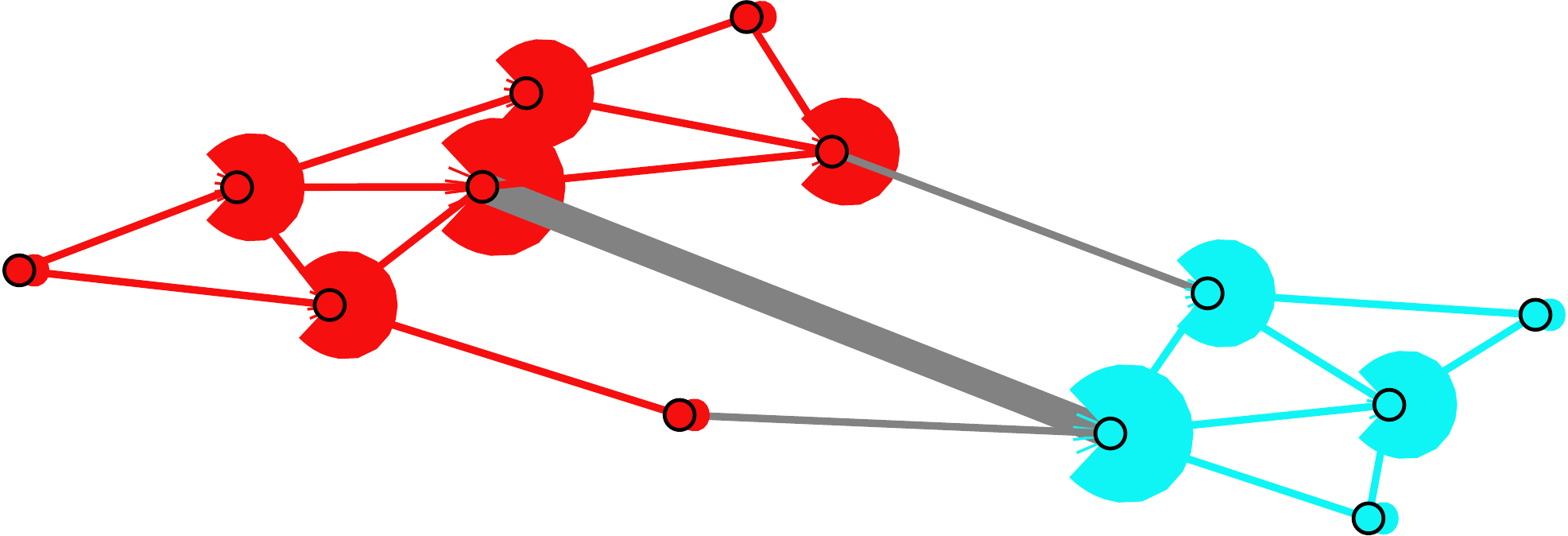}
    \caption{Traffic = $164$, Conductance = $0.12$}
  \end{subfigure}
  \end{center}
\caption{Communications network example, continued (see text). Two examples of bottleneck alleviation, using the results from \replaced{Fig.}{Figure} 1(b), by introducing a new edge, shown as the thick blue line in \deleted{3}(a)\deleted{,} and the thick red line in \deleted{3}(b).  While the total traffic measure is the same for both\added{ cases}, \deleted{case }(b) has higher conductance and leads to a more globally efficient network. In \deleted{3}(c), we \replaced{include the}{assume} additional delays \deleted{are }caused by the introduction of the new edge.}
\label{fig:hybrid}
\end{figure}
The new edges have a bandwidth that is four times greater than the preexisting edges (\added{and are }hence graphically thicker). We further assume the new edges do not introduce new interference to the preexisting edges, and so no new delays are incurred. The first case, shown \added{in} \figref{fig:hybrid}(a), connects two peripheral vertices at the corners of the two subsets, whereas the second case\added{,} shown in  \figref{fig:hybrid}(b)\added{,} connects two central vertices on opposite sides of the bottleneck.

In line with our intuition, the cross-link between central vertices in  \figref{fig:hybrid}(b) leads to an overall more efficient system as measured by conductance. The central vertices are more available to the entire network on both sides of the bottleneck. For both cases we also repeated the bottleneck discovery after the new edge was introduced, as shown in \figref{fig:hybrid} with the new color groupings. Comparing with  \figref{fig:wireless}(b), the new edge in \figref{fig:hybrid}(a) has less of an impact on the bottleneck subsets than the more effective new edge in  \figref{fig:hybrid}(b). 

To demonstrate that changing the network protocol on the same topology may result in different bottlenecks, we also repeated the bottleneck discovery with the assumption that the new \replaced{edge does}{edges do} introduce delays proportional to the increased vertex degrees. As shown in \figref{fig:hybrid}(c), the additional delays naturally lead to a less efficient system. Compared with \figref{fig:wireless}(b), the bottleneck location remains the same with the added edge\deleted{,} but \replaced{now with lower}{with improved} minimum conductance. 

This analysis can be expanded to discover bottlenecks and the corresponding optimal choices for new edges in an iterative manner, providing network enhancement options and adding robustness.

\section{Exploring graph frequency analysis for brain connectivity}
\label{sec:brain}
To illustrate how the Z-Laplacian based graph filters can be used in frequency analysis, we consider structural brain networks built from diffusion weighted imaging MRI scans\footnote{Diffusion weighted imaging captures bundles of white matter fibers, revealing the anatomical connections between different parts of the brain.} of 40 experiment participants \cite{betzel2013multi}. Frequency-specific brain activity is well known and associated with different brain states. Graph frequency analysis is thus a useful signal processing tool for studying functional brain networks \cite{Huang2016brain,FrequencyBrain}. Recently, there is also evidence that structural brain networks may also organize by graph spectrum \cite{daianu_spectral_2015,Medaglia2016functional,2016brainUSC}. However, previous work is mostly limited to the standard shift operators based on graph adjacency $\AA$ or (unnormalized) Laplacian matrices $\LL = \DD-\AA$. Here we demonstrate the flexibility of the Z-Laplacian in graph frequency analysis 
of given structural networks. In the future, we plan to \replaced{investigate}{do} direct signal analysis of functional networks. 

To demonstrate frequency analysis based on different candidate shift operators, we construct multiple Z-Laplacian operators using the same average network over all 40 samples. \figref{fig:structural} \replaced{depicts}{shows} the undirected weighted adjacency structure of the average structural network and a corresponding $2\%$ visualization.\footnote{An $x\%$ visualization displays the \replaced{$x\%$ of edges with the greatest weight}{top $x$ percentile edges in terms of the edge weights}. The percentile measure includes \replaced{zero-weight}{zero weight} edges or non-edges. In \figref{fig:structural}, the $2\%$ visualization contain the top 1,096 edges in terms of weight, whereas the $1\%$ visualizations in \figref{fig:filters} \replaced{each}{all} contain 548 edges.}
\begin{figure}
  \begin{subfigure}[b]{0.52\textwidth}
   \includegraphics[width=\textwidth]{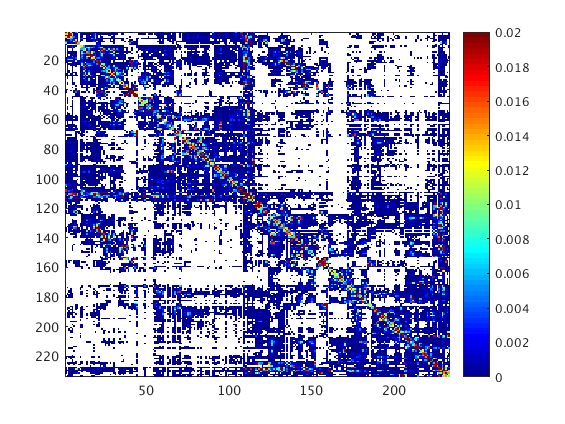}
  \end{subfigure}
  \begin{subfigure}[b]{0.42\textwidth}
   \includegraphics[width=\textwidth]{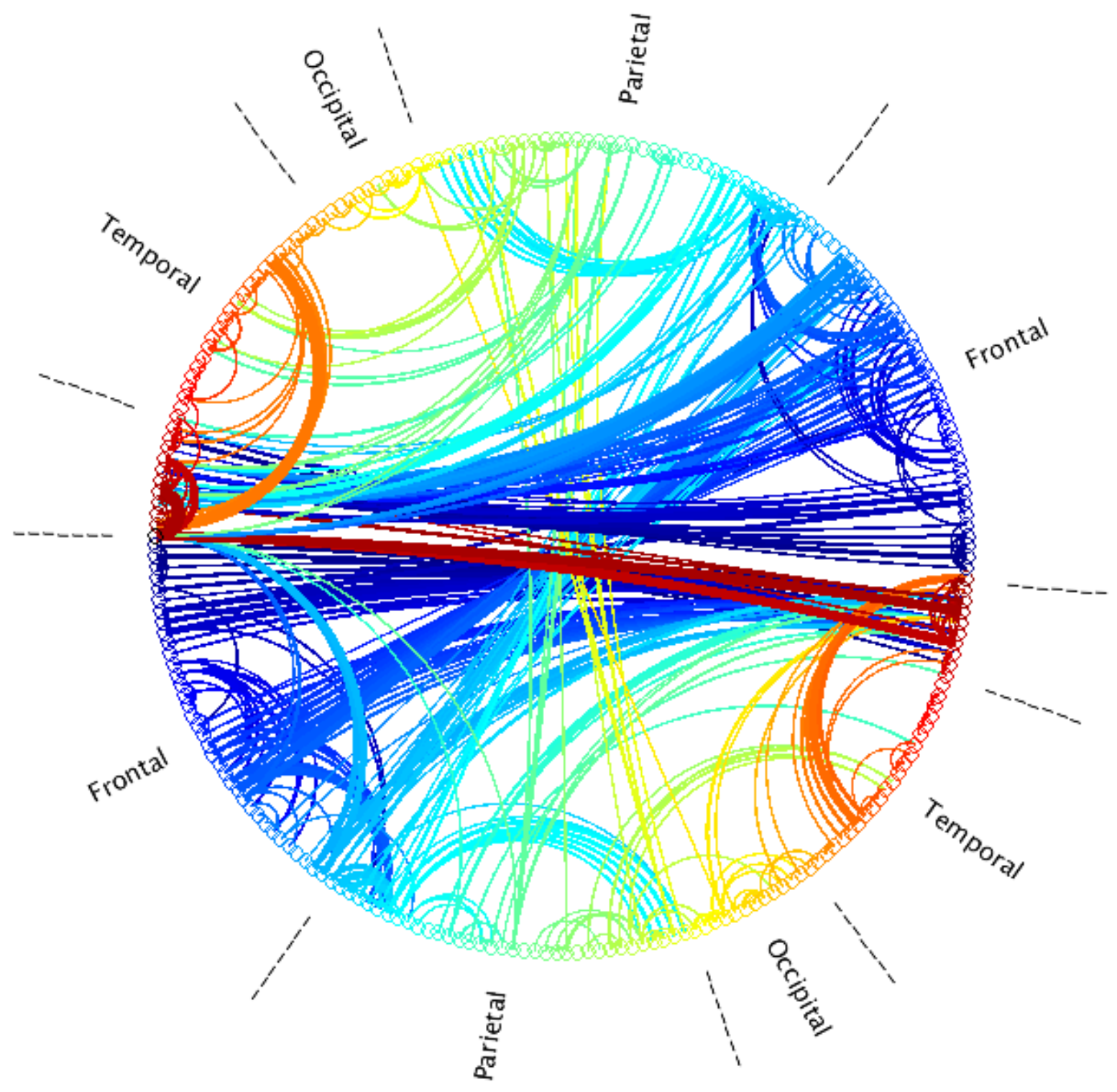}
  \end{subfigure}
\caption{The adjacency matrix and $2\%$ visualization of the average structural brain network}
\label{fig:structural}
\end{figure}
The visualization uses a circular layout similar to a connectogram \cite{irimia2012circular}. We have also labeled the four major regions of the cerebral cortex in the human brain, namely frontal lobe, parietal lobe, occipital lobe\added{,} and temporal lobe. For a better cross-hemisphere visualization, we adopted the circular symmetry and concatenated left and right hemispheres head to tail. We color coded the brain regions from blue to red in each hemisphere, and the stem is colored black. There are other brain regions in the network that are not labeled (dark red). Some vertices from the limbic lobe are merged into the frontal (blue) and parietal (green) lobes. Notice here the averaged network is quite dense with 27,540 total non-zero edges from all 40 samples.

Consider the following three Z-Laplacians as candidate shift operators:
\begin{align}
 \LL_0 &= {\DD}^{-1/2}(\DD - \AA){\DD}^{-1/2}\\
 \LL_1 &= \DD_{\WW}^{-1/2} [\DD_{\WW} - \DD^{-1}\AA\DD^{-1}] \DD_{\WW}^{-1/2}\\
 \LL_2 &= \ZZ^{1/2} \DD^{-1/2} (\ZZ^{-1}\DD - \AA) \DD^{-1/2}\ZZ^{1/2}\;.
\end{align}
These symmetric versions of \added{the }Z-Laplacian (Generalized Laplacians) are obtained through the similarity transformation (Lemma \ref{th:transformS}). \replaced{The first shift operator, $\LL_0$,}{$\LL_0$} is the symmetrized unbiased random walk Laplacian on the original adjacency. It is also the symmetric normalized Laplacian, which has been suggested as a shift operator \cite{shuman_signal_2012}. \replaced{The second shift operator, $\LL_1$,}{$\LL_1$} represents the biased random walk in \eqref{eq:degreeBias} with an increased tendency towards \replaced{lower-degree}{lower degree} vertices. The last shift operator\added{,} $\LL_2$\added{,} is the Z-Laplacian representing a self-replicating process based on the unbiased random walk $\LL_0$ but with $[\ZZ]_{ii}$ set to $10$ for vertices in the frontal lobe (blue) and $1$ for other vertices. Thus, $\LL_2$ emphasizes the hypothetical information flow associated with the frontal lobe. All three Z-Laplacian operators share the same uniform vertex time 
delays with $\TT=\II$.

Following the \replaced{approach}{idea} of graph frequency analysis \cite{shuman_signal_2012,sandryhaila_discrete_2014}, we consider the eigendecomposition\deleted[remark=rjd:made into equation with bold lambda]{}
\begin{align}
\LL = \VV \boldsymbol{\varLambda} \VV^T
\end{align}
of the Laplacians, where the columns of $\VV$ are eigenvectors of $\LL$ and \replaced{$\boldsymbol{\varLambda}$}{$\varLambda$} is a diagonal matrix comprising the corresponding eigenvalues. By replacing all but the four smallest eigenvalues with zero, we build simple low-pass filters based on the three different Z-Laplacian shift operators. The definition of \replaced{high- and low-frequency}{high/low frequency} bands is application specific. In this study, we choose the smallest four \replaced{eigenpairs}{eigen-pairs} to be the low frequencies because we are interested in analyzing the structure of four major brain regions. We can demonstrate the effect of these filters by reconstructing the adjacency structures based on the original $\DD$ and $\ZZ$ matrices:\deleted[remark=rjd: made into equation with bold lambda]{}
\begin{align}
\AA' = \ZZ^{-1/2} \DD^{1/2} (\II - \VV\boldsymbol{\varLambda}' \VV^T) \ZZ^{-1/2} \DD^{1/2}\;,
\end{align}
where \replaced{$\boldsymbol{\varLambda}'$}{$\varLambda'$} contains the first \replaced{four}{4} smallest eigenvalues on its diagonal, with all remaining elements set to zero.

The results are shown in \figref{fig:filters}, using $1\%$ visualizations. 
\begin{figure}
  \begin{subfigure}[b]{0.24\textwidth}
   \includegraphics[width=\textwidth]{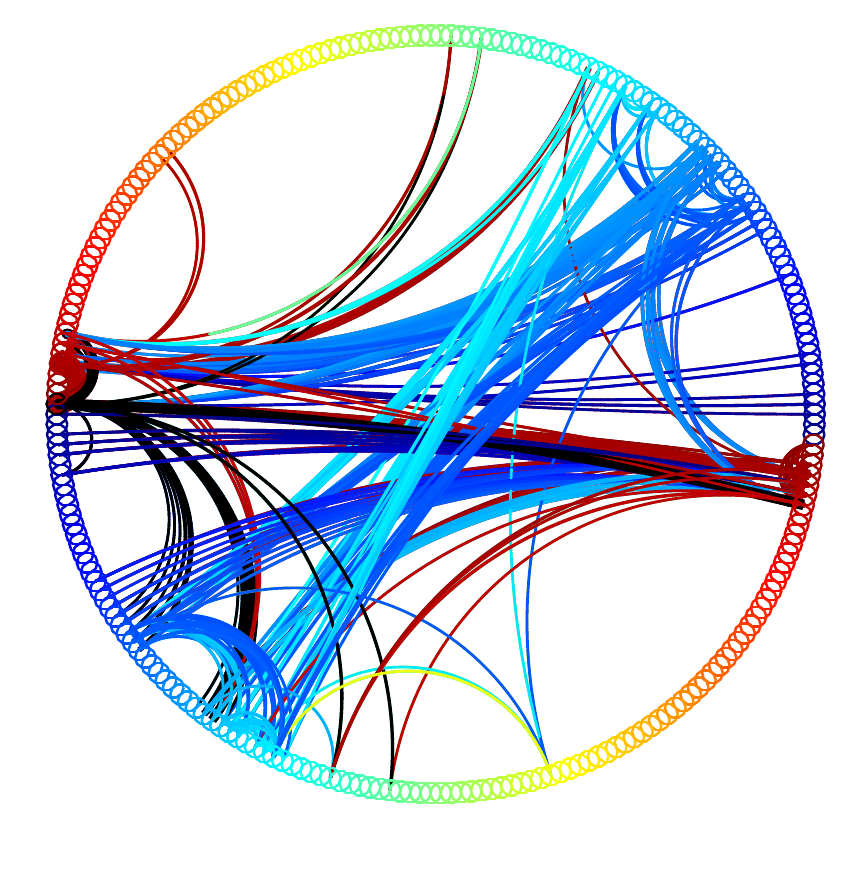}
   \caption{\replaced{Low-pass}{Low pass} filtered $\LL_0$}
  \end{subfigure}
  \begin{subfigure}[b]{0.24\textwidth}
   \includegraphics[width=\textwidth]{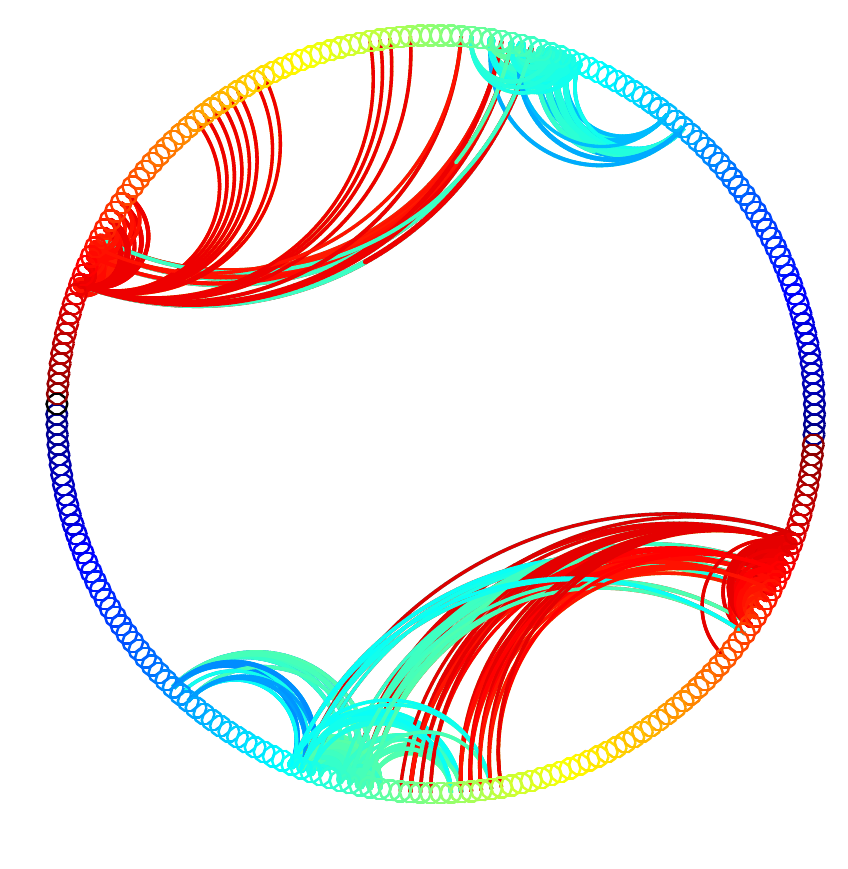}
   \caption{\replaced{Low-pass}{Low pass} filtered $\LL_1$}
  \end{subfigure}
  \begin{subfigure}[b]{0.24\textwidth}
   \includegraphics[width=\textwidth]{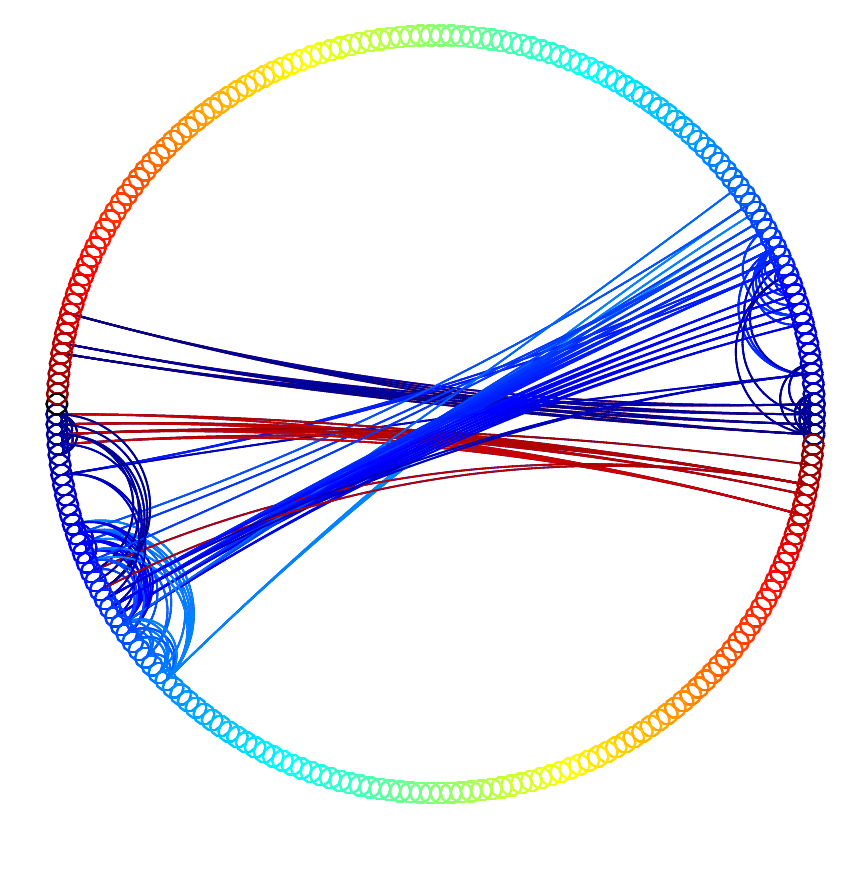}
   \caption{\replaced{Low-pass}{Low pass} filtered $\LL_2$}
  \end{subfigure}
  \begin{subfigure}[b]{0.24\textwidth}
   \includegraphics[width=\textwidth]{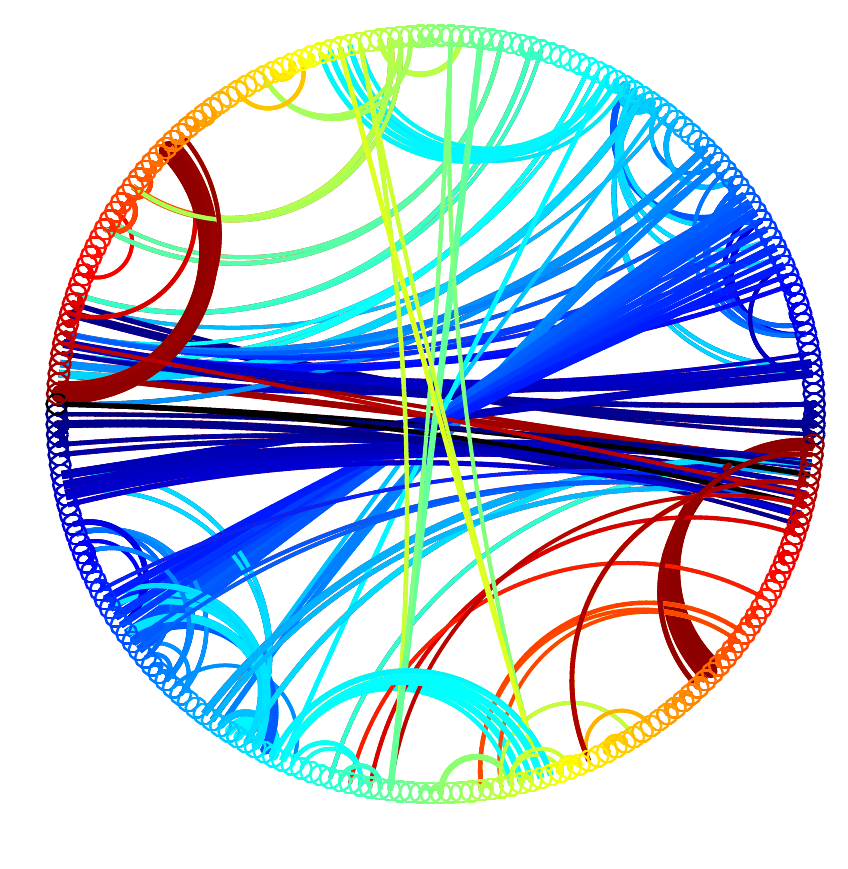}
   \caption{\replaced{High-pass}{High pass} filtered $\LL_0$}
  \end{subfigure}
\caption{\replaced{Depictions ($1\%$ visualizations) of the reconstructed adjacency structures after filtering.}{The reconstructed adjacencies after high-pass filtering with $1\%$ visualization}}
\label{fig:filters}
\end{figure}
Comparing \figref{fig:filters}(a) to \figref{fig:structural}, the \replaced{high-frequency}{high frequency} signals from $\LL_0$ are filtered out, leaving mostly stronger cross-hemisphere connections. The filtered $\LL_1$ network highlights distinctly different structures, with \replaced{low-degree}{low degree} vertices from \figref{fig:structural} now dominating the \replaced{low-frequency}{low frequency} spectrum. Strong cross-hemisphere connections associated with hub vertices are filtered out, revealing within-hemisphere connection patterns, especially those between parietal, occipital\added{,} and temporal lobes. In \figref{fig:filters}(c), the $\ZZ$ matrix is carefully designed to emphasize vertices in the frontal lobe. As expected, the result properly highlights the frontal lobe region and their internal 
connections, both within and across the hemispheres. 

Finally, \figref{fig:filters}(d) depicts the effect of applying a \replaced{high-pass}{high pass} filter obtained by replacing the four smallest eigenvalues of $\LL_0$ with zero and \replaced{leaving the rest unchanged}{keeping the rest}. Comparing Figs. 5(a)--(d), we see that the different shift operators each produce unique patterns.  Manipulating the Z-Laplacian allows exploration of different shift operators as we explore the brain connectivity. More importantly, the dynamical processes associated with each Z-Laplacian provides meaningful intuition and interpretations of the resulting shift operators and the induced family of linear invariant filters.

\section{Conclusions and future work}
In this paper, we proposed the \emph{Z-Laplacian} framework, which is capable of modeling different \replaced{discrete- and continuous-time}{discrete and continuous time} dynamical processes on graphs, including diffusion and epidemic processes. We proved that the Z-Laplacian spans the space of Z-matrices, leading to a general framework that unifies existing linear operators in the literature. When used as graph shift operators in applications, Z-Laplacian operators and their induced signal processing analysis have intuitive connections to the dynamical processes they model. This is especially useful \replaced{for relating and comparing different aspects of the same topological structures.}{when analyzing different aspects of the same topological structure, where they can be related and compared using interpretable parameters.}

The Z-Laplacian framework also naturally connects to concepts in network science, enabling graph theoretical methods to be used in signal processing problems. In particular, we showed a novel analysis that coupled network bottleneck discovery with the underlying wireless network protocol, including the impact of delay and collisions.  We demonstrated how conductance can be used to find primary bottlenecks, whose location and effect may change with choice of MAC protocol, and we considered topology modifications to alleviate the bottleneck and heal the network. This leads to more general questions of the effects of protocols on network dynamical processes, such as consensus, as well as the study of resource allocation within the network, which are important issues for further study. We also showed how a variety of graph shift operators can be applied to the problem of structural brain connectivity frequency analysis. In the future, we plan to apply the \replaced{GSP}{graph signal processing} tools to 
functional brain signals. We will also investigate the mathematical properties of Z-Laplacians, unifying vertex centrality and community structure under the framework as we did previously for the \replaced{\emph{parameterized Laplacians}}{\emph{parametrized Laplacians}} \cite{yan2016capturing}.



\begin{IEEEbiography}[{\includegraphics[width=1in,height=1.25in,clip,keepaspectratio]{figures/Yan}}]{Xiaoran Yan}
(M'17) received the B.S. degree from the Zhejiang University, China, in 2007 and the Ph.D. degree from the University of New Mexico, Albuquerque, NM, USA, in 2013, all in computer science. From 2013 to 2015, he was a Postdoctoral Research Associate at the University of Southern California Information Sciences Institute. He joined the Indiana University Network Science Institute in 2015 as an Assistant Research Scientist. His research interests include network science, statistical machine learning and their trans-disciplinary applications in social networks, communications networks and neuroscience. He has authored multiple journal and conference papers, and has participated as a reviewer or program committee member for major organizations such as IEEE and AAAI.
\end{IEEEbiography}

\begin{IEEEbiography}[{\includegraphics[width=1in,height=1.25in,clip,keepaspectratio]{figures/Sadler}}]{Brian M. Sadler}
(S'81--M'81--SM'02--F'07) received the B.S. and M.S. degrees from the University of Maryland, College Park, MD, USA, and the Ph.D. degree from the University of Virginia, Charlottesville, VA, USA, all in electrical engineering.  He is the US Army Senior Scientist for Intelligent Systems and a Fellow of the Army Research Laboratory (ARL) in Adelphi, MD, USA.  His research interests include information science and networked and autonomous intelligent systems.  He received Best Paper Awards from the Signal Processing Society in 2006 and 2010 and was general co-chair of the 2016 IEEE Global Conference on Signal and Information Processing. He was an associate editor for IEEE Transactions on Signal Processing and IEEE Signal Processing Letters.  He has been a guest editor for several journals, including IEEE JSTSP, IEEE JSAC, the IEEE SP Magazine, and the International Journal of Robotics Research.
\end{IEEEbiography}

\begin{IEEEbiography}[{\includegraphics[width=1in,height=1.25in,clip,keepaspectratio]{figures/Drost}}]{Robert J. Drost}
(M'10--SM'14) received the B.S. degree in electrical engineering from the University of Arkansas, Fayetteville, AR, USA, in 2000 and the M.S. degree in electrical engineering, the M.S. degree in mathematics, and the Ph.D. degree in electrical and computer engineering from the University of Illinois at Urbana-Champaign, Champaign, IL, USA, in 2002, 2005, and 2007, respectively. 

From 2007 to 2009, he was a Digital Signal Processing Design Engineer with Finisar Corporation. He joined the Army Research Laboratory (ARL) in 2010 as a Postdoctoral Fellow and, subsequently, an Electronics Engineer. He has authored several journal and conference papers, receiving the 2013 ARL Honorary Award for Publication. He has also authored four patents or patent applications and has participated as a reviewer and panelist for major organizations such as NSF and IEEE. His research interests include optical communications, signal processing, and graphical models. 
\end{IEEEbiography}

\begin{IEEEbiography}[{\includegraphics[width=1in,height=1.25in,clip,keepaspectratio]{figures/Yu}}]{Paul Yu}
(Member) received the Ph.D. degree in Electrical Engineering from the University of Maryland, College Park. Since 2006, he has been with the U.S. Army Research Laboratory (ARL) where his work is in the area of signal processing for wireless networking and autonomy. His most recent work focuses on the exploitation of mobility for improved wireless network connectivity in complex propagation environments. He received the Outstanding Invention of the Year award in 2008 and the Jimmy Lin Award for Innovation and Invention in 2009, both from the University of Maryland, and a Best Paper award at the 2008 Army Science Conference.
\end{IEEEbiography}

\begin{IEEEbiography}[{\includegraphics[width=1in,height=1.25in,clip,keepaspectratio]{figures/Lerman}}]{Kristina Lerman}
received the A.B. degree from Princeton University, Princeton, NJ, and the Ph.D. degree from the University of California at Santa Barbara, Santa Barbara, CA, all in Physics. She works as a Research Team Lead at the University of Southern California Information Sciences Institute and holds a joint appointment as a Research Associate Professor in the USC Computer Science Department. Trained as a physicist, she now applies network analysis and machine learning to problems in computational social science, including crowdsourcing, social network and social media analysis.  Her recent work on modeling and understanding cognitive biases in social networks has been covered by the Washington Post, Wall Street Journal, and MIT Tech Review.            
\end{IEEEbiography}

\section*{Appendix}
\subsection{Proof of Lemma \ref{th:transformB}}
\begin{lemma*}[Bias transformation]
Any biased random walk on $G = (V,E,\AA)$, with the diagonal matrix $\BB$ specifying vertex bias factors $b_v$, is equivalent to an unbiased random walk on the transformed graph $\WW = \AA \BB$. If $G$ is undirected, \replaced{then we instead consider }the transformed graph \deleted{is }$\WW = \BB \AA \BB$\added{ to maintain edges having equal weight in both directions}.
\end{lemma*}
\begin{proof}
Given the transformed graph $\WW$, the transition probability of an unbiased random walk going from vertex $u$ to vertex $v$ is defined as\deleted{,}\deleted[remark=rjd: modified equation]{}
\begin{equation}
\PP'_{uv} = [{\DD_{\WW}}_\text{out}^{-1}\WW]_{uv} = \dfrac{\WW_{uv}}{\sum_{v}\WW_{uv}} \propto b_v a_{uv}\;,
\end{equation}
which is equivalent to the transition probability \replaced{$P_{uv}^\text{BRW} \propto b_v a_{uv}$}{$P_{uv}^{BRW} \propto b_v a_{uv}$} of the biased random walk on the original adjacency matrix $\AA$.

For \added{an }undirected graph, we need to guarantee the transition probabilities are preserved \replaced{in}{on} both directions. Given $\WW = \BB \AA \BB$, we have\deleted[remark=rjd: modified equation]{}
\begin{align}
\PP'_{uv} &= [{\DD_{\WW}}^{-1}\WW]_{uv} = \dfrac{\WW_{uv}}{\sum_{v}\WW_{uv}} \propto \frac{b_v a_{uv} b_u}{b_u}\propto P_{uv}^\text{BRW}\nonumber\\
\PP'_{vu} &= [{\DD_{\WW}}^{-1}\WW]_{vu} = \dfrac{\WW_{vu}}{\sum_{u}\WW_{uv}} \propto \frac{b_v a_{vu} b_u}{b_v}\propto P_{vu}^\text{BRW}\;.
\end{align}
\end{proof}

\subsection{Proof of Lemma \ref{th:transformD}}
\begin{lemma*}[Delay transformation]
Any unbiased \replaced{continuous-time}{continuous time} random walk on $G = (V,E,\AA)$, with the diagonal matrix $\TT$ specifying vertex delay factors $\tau_v$, is equivalent to a \replaced{continuous-time}{continuous time} random walk with \deleted{the }delay factors $\II$ on the transformed graph \replaced{$\WW = \DD_\text{out}(\TT- \II)+\AA$}{$\WW = \DD_{out}(\TT- \II)+\AA$}.
\end{lemma*}
\begin{proof}
 \replaced{Starting}{Assuming we start} from the unbiased random walk Laplacian with the delay factors $\TT$\replaced{, we have}{:}\deleted[remark=rjd: modified equations]{}
\begin{align}
 \TT^{-1} & \DD_\text{out}^{-1}(\DD_\text{out}-\AA)= \II\TT^{-1} -  \TT^{-1}\DD_\text{out}^{-1}\AA\nonumber\\
				    &= \II -  \TT^{-1} (\TT- \II + \DD_\text{out}^{-1} \AA)\nonumber\\
				    &= {\DD_{\WW}}_\text{out}^{-1}{\DD_{\WW}}_\text{out}(\II - \TT^{-1}(\TT- \II) - \TT^{-1} \DD_\text{out}^{-1}\AA)\nonumber\\
				    &= {\DD_{\WW}}_\text{out}^{-1}({\DD_{\WW}}_\text{out} - \DD_\text{out}(\TT- \II) - \AA) \nonumber\\
				    &= \II{\DD_{\WW}}_\text{out}^{-1}({\DD_{\WW}}_\text{out}-\WW)\;,
\end{align}
where the new transformed matrix is \replaced{$\WW = \DD_\text{out}(\TT-\II)+\AA$}{$\WW = \DD_{out}(\TT- \II)+\AA$}\replaced{ and ${\DD_{\WW}}_\text{out}$}{, with ${\DD_{\WW}}_{out}$} represents its diagonal \replaced{out-degree}{out degree} matrix. 
\end{proof}

\subsection{Proof of Lemma \ref{th:discrete-continuous}}
\begin{lemma*}[Discrete approximation]
Given the graph signal $\ttheta(t)$ and the \replaced{continuous-time}{continuous time} Z-Laplacian \replaced{$\TT^{-1}(\II - \ZZ \DD_\text{out}^{-1}\AA)$}{$\TT^{-1}(\II - \ZZ \DD_{out}^{-1}\AA)$}, the graph signal at time $t+\delta$ can be approximated as $\delta\rightarrow 0$\replaced{ by}{,}
$$\lim_{\delta\rightarrow 0}\ttheta(t+\delta) =\ttheta(t) e^{-\LL\delta }  = \ttheta(t)(\II-\lim_{\delta\rightarrow 0}\delta\LL)\;,
$$
where $\II-\delta\LL$ represents a \replaced{discrete-time}{discrete time} filter.
\end{lemma*}

\begin{proof}
To better illustrate the role and impact of $\TT$ in the continuous version of the general nonnegative filter, we focus on the dynamics of a specific vertex $u$ over a time interval $\delta$, given by
\begin{align}
[\ttheta&(t+\delta) - \ttheta(t)]_u = [\ttheta(t) (e^{-\LL\delta }-\II)]_u \nonumber\\
&= -[\ttheta(t)]_u + \sum_{v} [\ttheta(t)]_v \left[ e^{-\lambda \delta} \sum_{k=0}^{\infty} \frac{1}{k!} (\lambda \boldsymbol{\varPhi} \delta)^k \right]_{vu}\;,
\end{align}
where $\boldsymbol{\varPhi} = \II-\frac{1}{\lambda}\LL$ is a \replaced{discrete-time}{discrete time} filter. For small $\delta$, we have\deleted[remark=rjd: modified equation]{}
\begin{align}
&[\ttheta(t+\delta) - \ttheta(t)]_u \nonumber\\
			&= -[\ttheta(t)]_u+ \sum_{v} [\ttheta(t)]_v \left[  e^{-\lambda \delta} (\II+\lambda \boldsymbol{\varPhi} \delta + O(\delta^2)) \right]_{vu} \nonumber\\
				&\approx -[\ttheta(t)]_u+ e^{-\lambda \delta}\sum_{v} [\ttheta(t)]_v \left[\II+\lambda (\II - \frac{1}{\lambda}\LL) \delta \right]_{vu}\nonumber\\
				&= \left( e^{-\lambda \delta} (1+\lambda\delta) -1\right) [\ttheta(t)]_u-e^{-\lambda \delta}\sum_{v} [\ttheta(t)]_v \left[\delta\LL\right]_{vu} \;,
\end{align}
where in the first equality we use the \emph{Big}-$O$ notation to capture the asymptotically shrinking quadratic and \replaced{higher-order}{higher order} terms of $\delta$. Therefore,\deleted[remark=rjd: modified equation]{}
\begin{align}
\label{eq:difference}
&\frac{[\ttheta(t+\delta) - \ttheta(t)]_u}{\delta}\nonumber\\
&\approx \left(\frac{e^{-\lambda \delta}-1}{\delta} +\lambda e^{-\lambda \delta}\right) [\ttheta(t)]_u-e^{-\lambda \delta}\sum_{v} [\ttheta(t)]_v \left[\LL\right]_{vu} \;.
\end{align}

It follows that\added{,} in the limit of $\delta\rightarrow0$, \deleted{Equation }\eqref{eq:difference} approaches the differential equation \eqref{eq:general_cont}\replaced{:}{,}
\begin{align}
\lim_{\delta\rightarrow 0}\frac{[\ttheta(t+\delta) - \ttheta(t)]_u}{\delta} = -\sum_{v} [\ttheta(t)]_v \left[\LL\right]_{vu} = \left[\frac{d \ttheta(t)}{dt}\right]_{u}\;.
\end{align}
\replaced{Rearranging terms}{If we rearrange the terms},\deleted[remark=rjd: modified equations]{}
\begin{align}
\lim_{\delta\rightarrow 0}[\ttheta(t+\delta)]_u	&= \lim_{\delta\rightarrow 0} [-\ttheta(t)\delta\LL + \ttheta(t)]_u\nonumber\\
						&= \lim_{\delta\rightarrow 0} [\ttheta(t)(1-\delta\LL)]_u\;,
\end{align}
and
\begin{align}
\lim_{\delta\rightarrow 0}\ttheta(t+\delta)= \ttheta(t)(\II-\lim_{\delta\rightarrow 0}\delta\LL)\;,
\end{align}
where the \replaced{discrete-time}{discrete time} filter $\II-\lim_{\delta\rightarrow 0}\delta\LL$ becomes a close approximation to the \replaced{continuous-time}{continuous time} filter.
\end{proof}

\subsection{Proof of Theorem \ref{th:interpretation}}
\begin{theorem*}[Delay interpretation]
Given the \replaced{continuous-time}{continuous time} Z-Laplacian \replaced{$\TT^{-1}(\II - \ZZ \DD_\text{out}^{-1}\AA)$}{$\TT^{-1}(\II - \ZZ \DD_{out}^{-1}\AA)$} and its \replaced{discrete-time}{discrete time} approximation $\boldsymbol{\varPhi}$, the delay factor \replaced{$[\TT]_{uu}$ of vertex $u$}{of vertex $u$, $[\TT]_{uu}$} is proportional to the expected ``waiting steps'' on vertex $u$ for the approximated \replaced{discrete-time}{discrete time} random walk.
\end{theorem*}
\begin{proof}
\replaced{We rewrite the discrete-time}{Rewriting the discrete time} filter $\boldsymbol{\varPhi}$\replaced{ as}{,}\deleted[remark=rjd: modified equation]{} 
\begin{align}
\label{eq:phi}
 \boldsymbol{\varPhi} &= \II - \delta \TT^{-1}(\II - \ZZ \DD_\text{out}^{-1}\AA)\nonumber\\
         &= (\II-\delta\TT^{-1}) + \delta \TT^{-1}\ZZ \DD_\text{out}^{-1}\AA\;.
\end{align}
Without loss of generality, we assume that \replaced{$\tau_i\le 1$ for all diagonal entries $\tau_i$}{$\tau_{i} \le 1 \forall \tau_{i}$, where the $\tau_i$ are the diagonal entries} of $\TT$. \replaced{Setting}{If we set} $\delta = \min_i \tau_{i}$, it follows that the second term of \deleted{Equation }\eqref{eq:phi} represents a row-wise \replaced{downscaling}{down-scaling} of the original \replaced{expression $\ZZ \DD_\text{out}^{-1}\AA$}{$\ZZ \DD_{out}^{-1}\AA$}, which can be rewritten as a new \replaced{discrete-time}{discrete time} filter given by\deleted[remark=rjd: modified equation]{}
\begin{equation}
 \delta \TT^{-1}\ZZ \DD_\text{out}^{-1}\AA = \ZZ' \DD_\text{out}^{-1}\AA\;.
\end{equation}

The first term of \eqref{eq:phi}, $\II-\delta\TT^{-1}$, with all entries in the real interval $[0,1)$, represents probabilities of self-looping random walks at each vertex. Since this is equivalent to staying at the same vertex with no dynamics, we can estimate the expected ``waiting steps'' at vertex $v$  to be\deleted{remark=rjd: modified equation} 
\begin{equation}
\sum_{k=0}^{\infty} \left(1-\frac{\delta}{\tau_v}\right)^k = \frac{1}{\delta/ \tau_v} = \frac{\tau_v}{\min_i \tau_i}\;,
\end{equation}
and so $\tau_v$ is the ``delay factor'' of vertex $v$.  
\end{proof}

\end{document}